\documentclass[11pt]{article}
\usepackage[utf8]{inputenc}
\usepackage[T1]{fontenc}
\usepackage{bm}
\usepackage{amsmath,amsfonts,amssymb,multirow}
\usepackage{amsthm}
\usepackage{multicol}
\usepackage{appendix}
\usepackage{cite}
\usepackage[margin=1in]{geometry}
\usepackage{fullpage} 
\usepackage{graphics}
\usepackage{float}
\usepackage[dvipsnames,usenames]{color}
\usepackage[colorlinks=true,urlcolor=Blue,citecolor=Green,linkcolor=BrickRed]{hyperref}
\usepackage[usenames,dvipsnames]{xcolor}
\usepackage{enumerate}
\usepackage{caption}
\usepackage{subcaption}

\usepackage{algorithm}
\usepackage[noend]{algorithmic}

\usepackage{paralist}
\usepackage{todonotes}

\usepackage{enumitem}

\newtheorem{lemma}{Lemma}
\newtheorem{theorem}{Theorem}

\newtheorem{conjecture}{Conjecture}

\makeatletter

\renewcommand{\theenumi}{\arabic{enumi}}

\renewcommand{\p@enumii}{\theenumi.}
\makeatother

\newcommand{\eps}{\varepsilon}
\newcommand{\cmatch}{c_{{}_{\mathrm{{match}}}}}
\newcommand{\cdel}{c_{{}_{\mathrm{{del}}}}}

\makeatletter

\begin{document}

\title{Tree Edit Distance Cannot be Computed in Strongly Subcubic Time (unless APSP can)}

\author{Karl Bringmann\thanks{Max Planck Institute for Informatics, Saarland Informatics Campus
    }
    \and
  Pawe\l{} Gawrychowski\thanks{University of Haifa. Partially supported by  the Israel Science Foundation grant 794/13.
  }
\and
 Shay Mozes\thanks{IDC Herzliya. Partially supported by  the Israel Science Foundation grant 794/13.
 }
\and
  Oren Weimann$^\dagger$
}

\date{}
\maketitle
\begin{abstract}
	The edit distance between two rooted ordered trees with $n$ nodes labeled from an alphabet~$\Sigma$ is the minimum cost of transforming one tree into the other by a sequence of elementary operations consisting of deleting and relabeling existing nodes, as well as inserting new nodes. 
Tree edit distance is a well known generalization of string edit distance. The fastest known algorithm for tree edit distance runs in cubic $O(n^3)$ time and is based on a similar dynamic programming solution as string edit distance. 
In this paper we show that a truly subcubic $O(n^{3-\varepsilon})$ time algorithm for tree edit distance is unlikely: For $|\Sigma| = \Omega(n)$, a truly subcubic algorithm for tree edit distance implies a truly subcubic algorithm for the all pairs shortest paths problem. 	For $|\Sigma| = O(1)$, a truly subcubic  algorithm for tree edit distance implies an $O(n^{k-\varepsilon})$ algorithm for finding a maximum weight $k$-clique. 

Thus, while in terms of upper bounds string edit distance and tree edit distance are highly related, in terms of lower bounds string edit distance exhibits the hardness of the strong exponential time hypothesis [Backurs, Indyk STOC'15] whereas tree edit distance exhibits the hardness of all pairs shortest paths. 
Our result provides a matching conditional lower bound for one of the last remaining classic dynamic programming problems.

\end{abstract}

\section{Introduction}

Tree edit distance is the most common similarity measure between labelled trees. Algorithms for computing the tree edit distance are being used in a multitude of applications in various domains including computational biology~\cite{Gus,Wat,Bille2005,ShapiroZ90}, structured text and data processing (e.g., XML)~\cite{FerraginaFOCS2005,BKG03,XML3}, programming languages and compilation~\cite{Hoffmann82}, computer vision~\cite{BellandoK99,Computervision}, character recognition~\cite{Rico-JuanM03}, automatic grading~\cite{AlurDGDV}, answer extraction~\cite{YaoDCC13}, and the list goes on and on.

Let $F$ and $G$ be two rooted
trees with a left-to-right order among siblings and where each
vertex is assigned a label from an alphabet $\Sigma$. The edit
distance between $F$ and $G$ is the minimum cost of transforming
$F$ into $G$ by a sequence of elementary operations (at most one operation per node): changing the label of a node $v$, deleting a node $v$ and setting the children of $v$ as the children of $v$'s parent (in the place of $v$ in the left-to-right order), and inserting a node $v$ (the complement of delete\footnote{Since a deletion in $F$ is
equivalent to an insertion in $G$ and vice versa, we can focus on
finding the minimum cost of a sequence of just deletions and
relabelings in both trees that transform $F$ and $G$ into isomorphic
trees.
}). See Figure~\ref{Fig:Operations}.  
The cost of these elementary operations is given by two
functions: $\cdel(a)$ is the cost of
deleting or inserting a vertex with label $a$, and
$\cmatch(a,b)$ is the cost of changing the label of a
vertex from $a$ to $b$. 

\begin{figure}[h!]
\begin{center}
\includegraphics[scale=0.43]{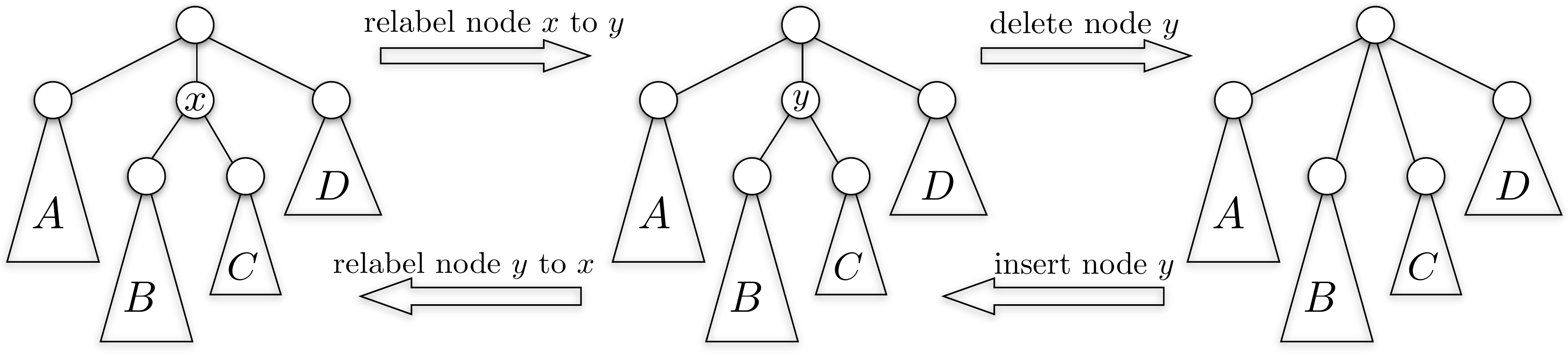}
\caption{\label{Fig:Operations} The three editing operations on a tree
with vertex labels.}
\end{center}
\end{figure}

The Tree Edit Distance (TED) problem was introduced by Tai in the late 70's~\cite{Tai} as a generalization of the well known string edit
distance problem~\cite{StringED}. Since then it was extensively studied.
Tai gave an $O(n^6)$-time algorithm for TED which was subsequently
improved to $O(n^4)$ in the late
80's~\cite{Shasha}, then to $O(n^3 \log n)$ in the late 90's~\cite{Klein}, and  finally to $O(n^3)$ in 2007~\cite{DMRW}. Many other algorithms have been developed for TED, see the popular survey of Bille~\cite{Bille2005} (this survey alone has more than 600 citations) and the books of Apostolico and Galil~\cite{ApostolicoBook} and Valiente~\cite{ValienteBook}. For example, Pawlik and Augsten~\cite{PawlikA} recently defined a class of dynamic programming algorithms that includes all the above algorithms for TED, and developed an algorithm whose performance on any input is not worse (and possibly better) than that of any of the existing algorithms.
Other attempts achieved better running time by restricting the  edit operations or the scoring schemes~\cite{XML3,Selkow,Constrained,Unit}, or by resorting to approximation~\cite{akutsu2,AratsuHK10}.
However, in the worst case no algorithm currently beats $\Omega(n^3)$ (not even by a logarithmic factor). 

Due to their importance in practice, many of the algorithms described above, as well as additional heuristics and  optimizations were studied experimentally~\cite{PawlikA,IvkinThesis}. 
Solving tree edit distance in truly subcubic $O(n^{3-\varepsilon})$ time is arguably one of the main open problems in pattern matching, and the most important one in tree pattern matching.

The fact that, despite the significant body of work on this problem, no truly subcubic time algorithm has been found, leads to the following natural conjecture that no such algorithm exists.

\begin{conjecture}\label{conj:TED}
For any $\epsilon >0$ Tree Edit Distance on two $n$-node trees cannot be solved in $O(n^{3-\epsilon})$ time.
\end{conjecture}

In the same paper proving the $O(n^3)$ upper bound for TED~\cite{DMRW}, Demaine et al. prove that their algorithm is optimal within a certain class of dynamic programming algorithms for TED. However, proving Conjecture~\ref{conj:TED} seems to be beyond our current lower bound techniques.

A recent development in theoretical computer science suggests a more fine-grained classification of problems in P. This is done by showing lower bounds conditioned on the conjectured hardness of certain archetypal problems such as  All Pairs Shortest Paths (APSP), 3-SUM, $k$-Clique, and Satisfiability, i.e., the Strong Exponential Time Hypothesis (SETH).

\paragraph{The APSP Conjecture.}
Given a directed or undirected graph with $n$ vertices and integer edge weights, many classical  algorithms for APSP (such as Dijkstra or Floyd-Warshall) run in $O(n^3)$ time.
The fastest  to date is the recent algorithm of  Williams~\cite{RyanWilliams} that runs faster than $O(n^3/\log^C n)$ time for all constants $C$.
Nevertheless, no truly subcubic $O(n^{3-\varepsilon})$ time algorithm for APSP is known.
This led to the following conjecture assumed in many papers, e.g.~\cite{BackursTzamos,BackursDikkalaTzamos,AbboudVassilevskaFOCS14,AbboudLewi2013,CountingWeightedSubgraphs,AbboudPlanar,RodittyZwick,AbboudGrandoniVassilevska,VW10,AmirVassilevskaYu}. 

\begin{conjecture}[APSP]\label{conj:APSP}
For any $\varepsilon > 0$ there exists $c > 0$, such that All Pairs Shortest Paths  on
$n$ node graphs with edge weights in $\{1,\ldots, n^c\}$ cannot be solved in $O(n^{3-\varepsilon})$ time.
\end{conjecture}

\paragraph{The (Weighted) \boldmath$k$-Clique Conjecture.}  
The fundamental $k$-Clique problem asks whether a given undirected unweighted graph on $n$ nodes and $O(n^2)$ edges contains a clique on $k$ nodes. 
This is the parameterized version of the famously NP-hard Max-Clique~\cite{Karp72}.
$k$-Clique is amongst the most well-studied problems in theoretical computer science, and it is the canonical intractable (W[1]-complete) problem in parameterized complexity.
A naive algorithm solves $k$-Clique in $O(n^k)$ time. A faster $O(n^{\omega k/3})$-time algorithm (where $\omega < 2.373$ is the exponent of matrix multiplication) can be achieved via a reduction to Boolean matrix multiplication on matrices of size $n^{k/3} \times n^{k/3}$ if $k$ is divisible by 3~\cite{NP85}\footnote{For the case that $k$ is not divisible by 3 see~\cite{EG04}.}.
Any improvement to this bound  immediately implies a faster algorithm for MAX-CUT~\cite{williams2005new,woeginger2008open}. It is a longstanding open question whether improvements to this bound are possible~\cite{woeginger,babai}.
The \emph{$k$-Clique conjecture} asserts that for no $k \ge 3$ and $\eps > 0$ the problem has an $O(n^{\omega k/3 - \eps})$ time algorithm, or an $O(n^{k-\eps})$ algorithm avoiding fast matrix multiplication, and has been used e.g. in~\cite{Clique,BringmannGL16}.

We work with a conjecture on a weighted version of $k$-Clique. 
In the {\em Max-Weight k-Clique} problem, the edges have integral weights and we seek the $k$-clique of maximum total weight. When the edge weights are small, one can obtain an $\tilde O(n^{k-\eps})$ time algorithm~\cite{Alon:1997,NP85}. However, when the weights are larger than $n^k$, the trivial $O(n^k)$ algorithm is the best known (ignoring $n^{o(k)}$ improvements).
This gives rise to the following conjecture, which has been used e.g. in \cite{AbboudVassilevskaWeimann,BackursTzamos,BackursDikkalaTzamos}.

\begin{conjecture}[Max-Weight $k$-Clique] \label{conj:maxclique}
For any $\varepsilon>0$ there exists a constant $c>0$, such that for any $k\geq 3$ 
Max-Weight $k$-Clique
on $n$-node graphs with edge weights in $\{1,\ldots,n^{ck}\}$ cannot be solved in
$O(n^{k(1-\varepsilon)})$ time.
\end{conjecture}

In 2014, with the burst of the conditional lower bound paradigm, Abboud~\cite{YRICALP} highlighted seven main open problems in the field: The first two were to prove quadratic $n^{2-o(1)}$ lower bounds for  {\em String Edit Distance} and {\em Longest Common Subsequence}, which were soon resolved in STOC'15~\cite{BackursIndyk} and FOCS'15~\cite{Bringmann,AbboudLCS} conditional on SETH. The third problem was to show a cubic $n^{3-o(1)}$ lower bound for {\em RNA-Folding}. Surprisingly, in  FOCS'16~\cite{RNA} it was shown that RNA-Folding can actually be solved in truly subcubic time, thus ruling out the possibility of such a lower bound. 
The remaining four problems remain open. In fact, two of them, showing a cubic lower bound for {\em Graph Diameter}  and an  $n^{\lceil k/2 \rceil-o(1)}$ lower bound for {\em k-SUM}, have actually been used as hardness conjectures themselves, e.g., in SODA'15~\cite{AbboudGrandoniVassilevska} and ICALP'13~\cite{AbboudLewi2013}. Until the present work, no progress has been made on the last problem posed by Abboud: A cubic lower bound for {\em Tree Edit Distance}. In the absence of progress on either upper bounds or conditional lower bounds for TED, one might think that Conjecture~\ref{conj:TED} is yet another fragment in the current landscape of fine grained complexity, and is unrelated to other common conjectures.

\subsection{Our Results}

In this paper we resolve the complexity of tree edit distance by showing a tight connection between edit distance of trees and all pairs shortest paths of graphs. We prove that Conjecture~\ref{conj:APSP} implies Conjecture~\ref{conj:TED}, and that Conjecture~\ref{conj:maxclique} implies Conjecture~\ref{conj:TED}, even for constant alphabet.

\begin{theorem} \label{thm:main}
	A truly subcubic algorithm for tree edit distance on alphabet size $|\Sigma| = \Omega(n)$ implies a truly subcubic algorithm for APSP. 	A truly subcubic  algorithm for tree edit distance on sufficiently large alphabet size $|\Sigma| = O(1)$ implies an $O(n^{k(1-\varepsilon)})$ algorithm for Max-Weight $k$-Clique.
\end{theorem}

Note that the known upper bounds for {\em string} edit distance and {\em tree} edit distance are highly related. The $O(n^2)$ algorithm for strings and the $O(n^3)$ algorithm for trees (and forests) are both based on a similar recursive solution: The recursive subproblems in strings (forests) are obtained by either deleting, inserting, or matching the rightmost or leftmost character (root). In strings, it is best to always consider the rightmost character. The recursive subproblems are then prefixes and the overall running time is $O(n^2)$. In trees however, sticking with the rightmost (or leftmost) root may result in an $O(n^4)$ running time. The specific way in which the recursion switches between leftmost and rightmost roots is exactly what enables the $O(n^3)$ solution. It is interesting that while the upper bounds for both problems are so similar, the lower bounds string edit distance exhibits the hardness of the SETH while tree edit distance exhibits the hardness of APSP.    

While a considerable share of the recent conditional lower bounds is on string pattern matching problems~\cite{BringmannGL16,CIP09,larsen2015hardness,Clique,backurs2016regular,AbboudVassilevskaWeimann,BackursIndyk,polylogshaved,AbboudLCS,JumbledHardness,Bringmann}, the only conditional lower bound for a tree pattern matching problem is the recent SODA'16 quadratic lower bound for {\em exact pattern matching}~\cite{AmirIsomorphism} (the problem of deciding whether one tree is a subtree of another). We solve the main remaining open problem in tree pattern matching, and one of the last remaining classic dynamic programming problems. 
Indeed, apart from the problems discussed above, for most of the classic dynamic programming problems a conditional lower bound or an improved algorithm have been found recently. 
This includes the Fr\'echet distance~\cite{bringmann2014walking}, bitonic TSP~\cite{deberg_et_al}, context-free grammar parsing~\cite{Clique}, maximum weight rectangle~\cite{BackursDikkalaTzamos}, and pseudopolynomial time algorithms for subset sum~\cite{bringmann2017near}.
Tree edit distance was one of the few classic dynamic programming problems that so far resisted this approach. Two notable remaining dynamic programming problems without matching bounds are the optimal binary search tree problem ($O(n^2)$)~\cite{knuth1971optimum} and knapsack (pseudopolynomial $O(n W)$)~\cite{bellman1957dynamic}.

\subsection{Our Reductions}

\paragraph{APSP to TED.}
In order to prove APSP-hardness, by~\cite{VW10} it suffices to show a reduction from the negative triangle detection problem, where we are given an $n$-node graph $G$ with edge weights $w(.,.)$ and want to decide whether there are $i,j,k$ with $w(i,j) + w(j,k) + w(i,k) < 0$. 
Our first result is a reduction from negative triangle detection to tree edit distance, which produces trees of size $O(n)$ over an alphabet of size $O(n)$. This yields the matching conditional lower bound of $O(n^{3-\eps})$.

Our reduction constructs trees that are of a very special form: Both trees consist of a single path (called spine) of length $O(n)$ with a single leaf pending from every node (see Figure~\ref{fig:macrogeneral}).
Such instances already have been identified as difficult for a 
restricted class of algorithms based on a specific dynamic programming approach~\cite{DMRW}. In our setting we
cannot assume anything about the algorithm, and hence need a deeper insight on the structure of any valid sequence of edit operations (see Figure~\ref{fig:macrogeneral} and Lemma~\ref{lem:redstructure}).   
Using this structural understanding, we then show that it is possible to carefully construct a cost function so that any optimal solution must obey a certain structure (Figure~\ref{fig:macrooptimal}). Namely, for some $i,j,k$ we match the two leaves in depth $k$, we match the right spine node in depth $k$ to the left leaf in depth $i$ (which encodes $w(i,k)$), we match the left spine node in depth $k$ to the right leaf in depth $j$ (which encodes $w(j,k)$), and we match as many spine nodes above depth $i$ and $j$ as possible (which together encode $w(i,j)$ by a telescoping sum).

\paragraph{Constant alphabet size.}
The drawback of the above reduction is the large alphabet size $|\Sigma|$, as essentially each node needs its own alphabet symbol.
There are two major difficulties to improving this to constant alphabet size. 

First, the instances identified as hard by the above reduction (and by Demaine et al.~\cite{DMRW} for a restricted class of algorithms) are no longer hard for small alphabet! Indeed, in Section~\ref{sec:algo} we give an $O(n^{2}|\Sigma|^{2} \log n)$ algorithm for these instances, which is truly subcubic for constant alphabet size. This algorithm is the first to break the barrier by Demaine et al.\ for such trees, and we believe it is of independent interest. Regarding lower bounds, this algorithm shows that for a reduction with constant alphabet size our trees necessarily need to be more complicated, making it much harder to reason about the structure of a valid edit sequence. We will construct new hard instances by taking the previous ones and attaching small subtrees to all nodes.

The second difficulty is that, since the input size of TED 
is $\tilde O(n + |\Sigma|^2)$, a reduction from negative triangle detection to TED with constant alphabet size would need to considerably compress the $\Omega(n^2)$ input size of negative triangle detection. It is a well-known open problem whether such compressing reductions exist.
To circumvent this barrier, we assume the stronger Max-Weight $k$-Clique
Conjecture, where the input size $\tilde O(n^2)$ is very small compared to the running time $O(n^k)$. 

\paragraph{Max-Weight \boldmath$k$-Clique to TED.}
Given an instance of Max-Weight $k$-Clique on an $n$-node graph $G$ and weights bounded by $n^{O(k)}$ we construct a TED instance on trees of size $O(n^{k/3+2})$ over an alphabet of size $O(k)$. One can verify that an $O(n^{3-\eps})$ algorithm for TED now implies an $O(n^{k(1-\eps/6)})$ algorithm for Max-Weight $k$-Clique, for any sufficiently large $k=k(\eps)$.  

We roughly follow the reduction from negative triangle detection; now each spine node corresponds to a $k/3$-clique in $G$. 
To cope with the small alphabet, we simulate the previous matching costs with small gadgets. In particular,
to each spine node, corresponding to some $k/3$-clique $U$, we add a small subtree $T(U)$ of size $O(n^2)$ such that the edit distance between $T(U)$ and $T(U')$ encodes the total weight of edges between $U$ and $U'$. 
This raises two issues. First, we need to represent a
weight $w \in \{0,\ldots,n^{O(k)}\}$ by trees over an alphabet of size $O(k)$ (that is, constant).
This is solved by writing $w$ in base $n$ as $\sum_{i=0}^{O(k)} \alpha_i n^i$ and
constructing $\alpha_i$ nodes of type $i$, such that the cost of matching two type $i$ nodes is $n^{i}$. 
A second issue is that we need the small subtree $T(U)$
to interact with every other small subtree $T(U')$. So, in a sense,
$T(U)$ needs to ``prepare'' for any possible $U'$, and yet
its size needs to be small. We achieve this by creating in $T(U')$, for every node $u$ in $G$, a separate component responsible for counting the total weight of all edges between
$u$ and all nodes in $U'$. Then, in $T(U)$
we have a separate component for every node $u\in U$, and make sure that it is necessarily
matched to the appropriate component in $T(U')$.

The final and most intricate component of our reduction is to enforce that in any optimal solution
we have some control on which small subtrees can be matched to which.
A similar issue was present in the negative triangle reduction,  
when we require control over which spine nodes above depth $i$ are matched to which spine nodes above depth $j$. 
This is handled in the negative triangle reduction by assigning a different matching cost depending on the node's depth.
Now however, we cannot afford so many different
costs. We overcome this with yet another gadget, called an $I$-gadget, 
that achieves roughly the same goal, but in a more ``distributed'' manner.

Both of our reductions are highly non-trivial and introduce a number of new tricks that could be useful for other problems on trees.

\section{Reducing APSP to TED}
\label{sec:apsp}

We re-define the cost of matching two nodes to be the original cost minus the cost of deleting
both nodes. Then, the goal of TED amounts to choosing a subset of \emph{red nodes} in both trees,
so that the subtrees defined by the red nodes are isomorphic (i.e., their left-right and ancestor-descendant relation
is the same in both trees) and the total cost of matching the corresponding red nodes is minimized. See Figure~\ref{fig:macrogeneral}.
We work with this formulation from now on.

\begin{figure}[h!]
\begin{center}
\includegraphics[scale=0.4]{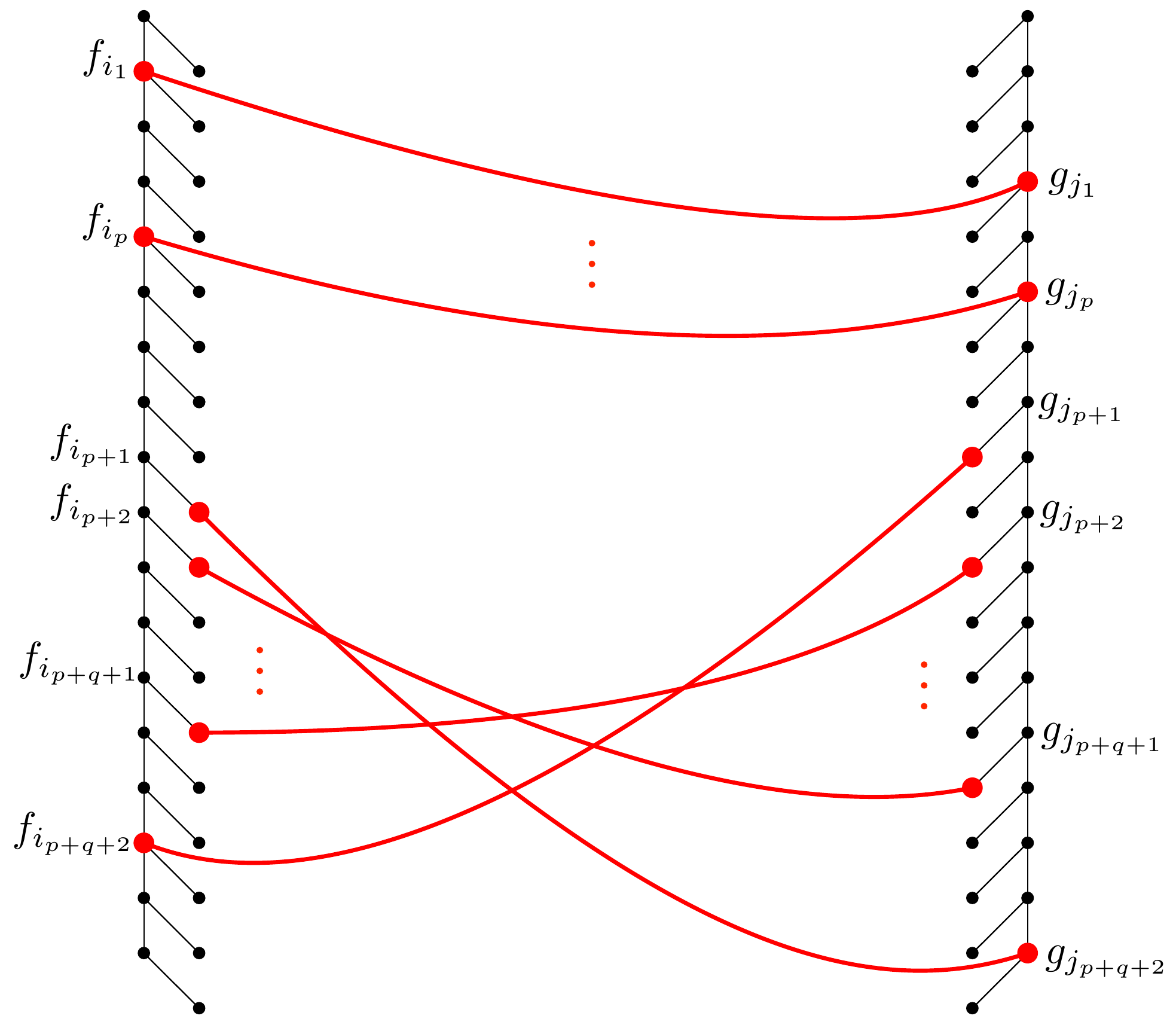}
\caption{Macro structure of the hard instance for TED: A tree $F$ composed of a single spine with leaves hanging to the right and a tree $G$ composed of a single spine with leaves hanging to the left.\label{fig:macrogeneral}}
\end{center}
\end{figure}

It turns out that a hard instance for TED is given by two seemingly simple caterpillar trees. These two trees $F$ and $G$, also called left and right, are shown in Figure~\ref{fig:macrogeneral}. 
Each tree consists of \emph{spine nodes} and \emph{leaf nodes}. If $u$ is a spine node then we denote by $u'$ the (unique)
leaf node attached to $u$. For any such hard instance of TED,
the red nodes in any matching have the structure given by Lemma~\ref{lem:redstructure} below. Informally, it states that starting from the top of the left tree and ordering the nodes by depth, the matching consists of (1) a prefix of a matching subsequence of spine nodes in both trees, (2) a suffix of a matching subsequence of leaf nodes that are in reverse order in the other tree, and (3) at most one final spine node in each of the trees matching a leaf node in the other tree that is located between the prefix part (1) and the suffix part (2).

\begin{lemma}
\label{lem:redstructure}
Let $f_{1},f_{2},\ldots$ and $g_{1},g_{2},\ldots$ denote the spine nodes of $F$ and $G$, respectively,
ordered by the depth. Then, for some $p,q\geq 0$ and some $i_{1}<i_2\cdots < i_{p}<i_{p+1}<\cdots <i_{p+q+1}<i_{p+q+2}$ and $j_{1}<j_2\cdots < j_{p}<j_{p+1}<\cdots <j_{p+q+1}<j_{p+q+2}$  the set of red nodes consists of:
\begin{enumerate}[leftmargin=0.6cm, label=(\arabic*)]
\item Spine nodes $f_{i_{1}},f_{i_{2}},\ldots,f_{i_{p}}$ matched respectively to spine nodes $g_{j_{1}},g_{j_{2}},\ldots,g_{j_{p}}$,\label{redstructure1}
\item Leaf nodes $f'_{i_{p+2}},f'_{i_{p+3}},\ldots,f'_{i_{p+q+1}}$ matched respectively to leaf nodes $g'_{j_{p+q+1}},g'_{j_{p+q}},\ldots,g'_{j_{p+2}}$ (note the reversed order),\label{redstructure2}
\item Optionally,  a spine node $f_{i_{p+q+2}}$ matched to leaf node $g'_{j_{p+1}}$. Also optionally, a spine node $g_{j_{p+q+2}}$ matched to a leaf node $f'_{i_{p+1}}$. \label{redstructure3}
\end{enumerate}

\end{lemma}

\begin{proof}
Consider the subtree defined by the red nodes. It has two isomorphic copies, one in $F$ and one in $G$. Its nodes are all the red nodes. The children of node $u$ are all red nodes $v_{1},v_{2},\ldots,v_{k}$ whose lowest red ancestor is $u$. The order is such that $v_i$ precedes $v_{i+1}$ in a left-to-right preorder traversal of $F$ (or equivalently of $G$). Let $u$ be a red node with two or more children $v_{1},v_{2},\ldots,v_{k}$, $k\geq 2$. Observe that $u$ must correspond to spine nodes in both $F$ and
$G$. Further observe that at most one $v_i$ can correspond to a spine node (otherwise, for two spine nodes one must be an ancestor of the other). 
Consider any $\ell\in\{1,2,\ldots,k-1\}$. It is not hard to see that node $v_{\ell+1}$ must correspond to a leaf node
in $F$ and node $v_{\ell}$ must correspond to a leaf node in $G$. This implies that both
$v_{\ell}$ and $v_{\ell+1}$ are leaves in the red subtree. Moreover, $v_1$ is the only node that may correspond to a spine node in $F$ and $v_k$ is the only node that may correspond to a spine node in $G$.
Consequently, the red subtree has
a particularly simple structure: it consists of nodes $u_{1},u_{2},\ldots,u_{p}$ such that for every $\ell=1,2,\ldots,p-1$ the
only child of $u_{\ell}$ is $u_{\ell+1}$, and nodes $v_{1},v_{2},\ldots,v_{k}$ (for some $k\geq 1$) that are all children
of $u_{p}$.

For every $\ell=1,2,\ldots,p$, the node $u_\ell$ must correspond to a spine node $f_{i_{\ell}}\in F$ and
$g_{j_{\ell}}\in G$. We immediately obtain~\ref{redstructure1} that $i_{1}<i_{2}<\ldots<i_{p}$ and that $j_{1}<j_{2}<\ldots<j_{p}$.
The nodes $v_{1},v_{2},\ldots,v_{k}$ are all children of $u_{p}$ in the subtree. It is possible that all $v_{i}$ are mapped to leaf nodes $f'_{i_{p+2}},f'_{i_{p+3}},\ldots,f'_{i_{p+q+1}}$ and  $g'_{j_{p+2}},g'_{j_{p+3}},\ldots,g'_{j_{p+q+1}}$. In this case, they must be mapped in reverse order since a left-to-right preorder traversal visits the leaves of $G$ in order of their depth and in reverse-depth order in $F$. This implies~\ref{redstructure2} that $i_{p}\leq i_{p+2}< \ldots < i_{p+q+1}$ and $j_{p} \le j_{p+2} < \ldots < j_{p+q+1}$. Recall however that $v_{1}$ may be mapped to a spine node $f_{i_{p+q+2}}$ in $F$ and a leaf node $g'_{j_{p+1}}$ in $G$. The requirement that $i_{p+q+2}>i_{p+q+1}$ and that  $j_{p+2} > j_{p+1}\ge j_p$ follows from the fact that these nodes correspond to a leftmost leaf in the subtree.
For symmetric reasons, $v_{k}$ may be matched to a spine
node $g_{j_{p+q+2}}\in G$ for some $j_{p+q+2}>j_{p+q+1}$ and $i_{p+2} > i_{p+1}\ge i_p$. This implies~\ref{redstructure3} and concludes the proof.
\end{proof}

The above lemma characterizes the structure of a solution to what we call the {\em hard instance} of TED. 
We next show how to reduce 
the \emph{negative triangle detection} problem to TED on the hard instance. Negative triangle detection is known to be subcubic equivalent to APSP~\cite{VW10}. Given a complete weighted
$n$-node undirected graph, where $w(i,j)$ denotes the weight of the edge $(i,j)$, the problem asks whether there are $i,j,k$ such that
$w(i,j)+w(j,k)+w(i,k) < 0$. To solve negative triangle detection, we clearly only need to find $i,j,k$ that minimize
$w(i,j)+w(j,k)+w(i,k)$. 
We will show how to construct, given such a graph, a hard instance of TED of size $O(n)$,
such that $\min_{i,j,k}w(i,j)+w(j,k)+w(i,k)$ can be extracted from the edit distance. 

\begin{lemma}
Given a complete undirected $n$-node graph $G$ with weights $w(.,.)$
in $\{1,\ldots,n^{c}\}$, we construct, in linear time in the output size, an instance of TED of size $O(n)$ with alphabet size $|\Sigma|=O(n)$
such that the minimum weight of a triangle in $G$ can be extracted from the edit distance. 
\end{lemma}

Consequently,
an $O(n^{3-\epsilon})$ time algorithm for TED implies an $O(n^{3-\epsilon})$ algorithm for negative triangle detection, and thus an $O(n^{3-\epsilon/3})$ algorithm for APSP by a reduction in~\cite{VW10}.

\begin{figure}[h!]
\begin{center}
\includegraphics[scale=0.7]{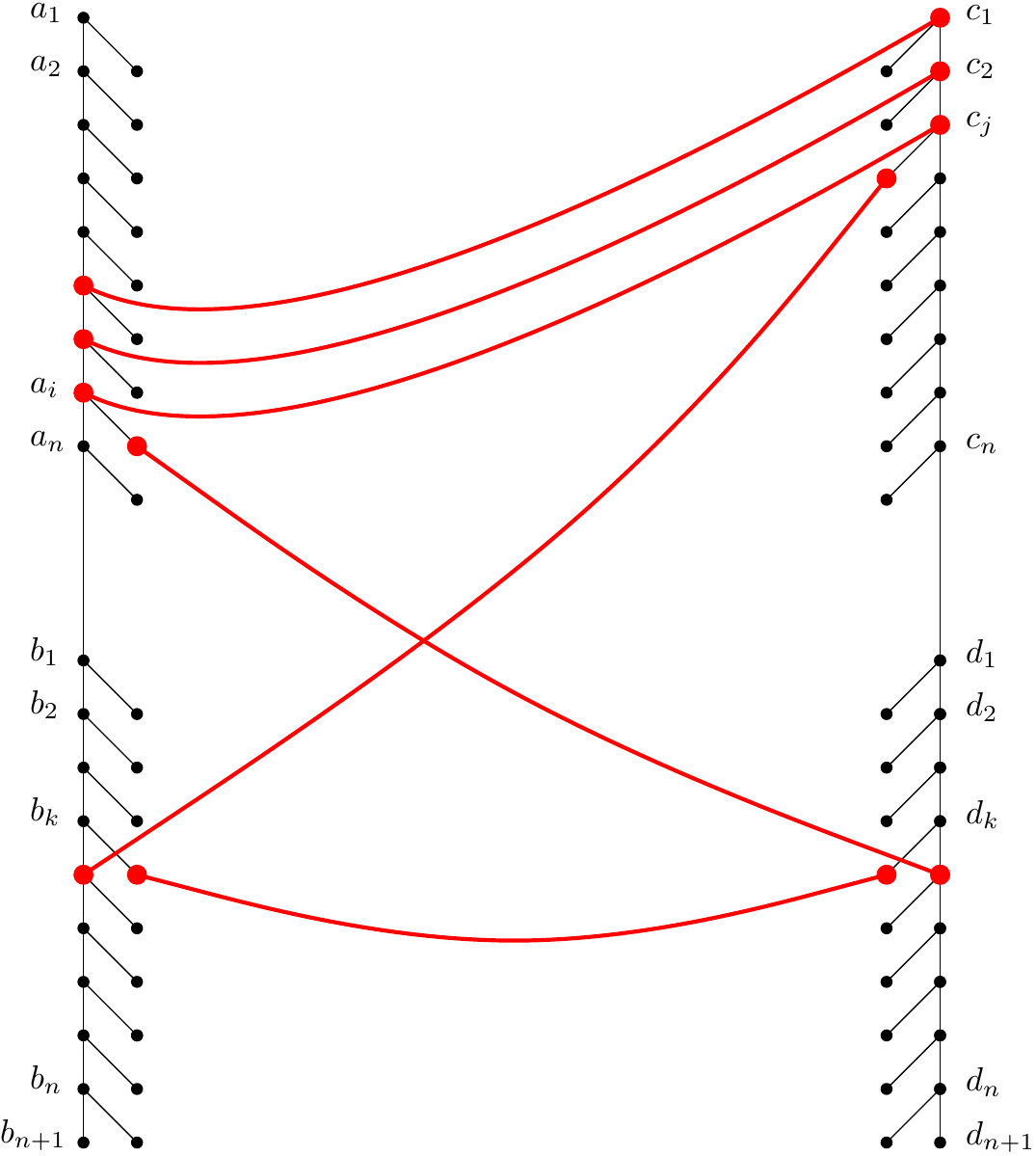}
\caption{A hard instance of TED constructed for a given instance of negative triangle detection. Appropriately
chosen costs ensure that any optimal solution has a specific structure. \label{fig:macrooptimal}}
\end{center}
\end{figure}

We create a hard instance of TED consisting of two trees $F$ and $G$ as in Figure~\ref{fig:macrooptimal}.
Every tree is divided into a \emph{top}
and a \emph{bottom} part. The spine nodes of these parts are denoted by $a_{1},a_{2},\ldots,a_{n}$ for
the top left part, $b_{1},b_{2},\ldots,b_{n+1}$ for the bottom left part, $c_{1},c_{2},\ldots,c_{n}$ for the top
right part, and $d_{1},d_{2},\ldots,d_{n+1}$ for the bottom right part.
The labels of all nodes are distinct (hence the alphabet size $|\Sigma|$ is $\Theta(n)$). 
We set the cost $\cmatch(u,v)$ of matching two nodes $u$ and $v$ as described
below, where $M$ denotes a sufficiently large number to be specified later. Intuitively, our assignment of costs ensures that any valid solution to TED must match
$b'_{k}$ to $d'_{k}$, $b_{k+1}$ to $c'_j$, and $d_{k+1}$ to  $a'_{i}$ for some $i,j,k$ (as shown in Figure~\ref{fig:macrooptimal}). Furthermore, the optimal solution (i.e., of minimum cost) must choose $i,j,k$ that minimize $w(i,k)+w(k,j)+w(i,j)$. The costs are assigned as follows:

\begin{enumerate}[label=(\arabic*)]
\item $\cmatch(b'_{k},d'_{k})=-M^{2}-2M\cdot k$ for every $k=1,2,\ldots, n$. \label{match:1}
\item $\cmatch(b_{k+1},c'_{j})=-M^{2}+M\cdot k+M\cdot j+w(k,j)$ for every $k=1,2,\ldots,n$ and $j=1,2,\ldots ,n$. \label{match:2}
\item $\cmatch(a'_{i},d_{k+1})=-M^{2}+M\cdot k+ M\cdot i+w(i,k)$ for every $i=1,2,\ldots,n$ and $k=1,2,\ldots ,n$. \label{match:3}
\item $\cmatch(a_{i},c_{j})=-2M+w(i,j)-w(i-1,j-1)$ for every $i=2,3,\ldots,n$ and $j=2,3,\ldots ,n$. \label{match:4}
\item $\cmatch(a_{i},c_{1})=-M(i+1)+w(i,1)$ for every $i=1,2,\ldots,n$. \label{match:5}
\item $\cmatch(a_{1},c_{j})=-M(j+1)+w(1,j)$ for every $j=1,2,\ldots,n$. \label{match:6}
\end{enumerate}

All the remaining costs $\cmatch(u,v)$ are set to $\infty$. The following theorem proves that these costs imply the required structure on the optimal solution as described above. 
Intuitively, by choosing  sufficiently large $M$, because of the $-M^{2}$ addend in \ref{match:1}, \ref{match:2} and \ref{match:3} we can ensure that any optimal solution matches
$b'_{k}$ to $d'_{k}$, $b_{k'+1}$ to $c'_j$, and $d_{k''+1}$ to $a'_{i}$, for some $i,j,k$ and
$k',k'' \leq k$.
Then, because of the $M\cdot k$ in \ref{match:2} and \ref{match:3}, in any optimal solution
actually $k=k'=k''$ and the total cost of all these matchings is $w(k,j)+w(i,k)$. Finally, because of the $-2M$ in \ref{match:4}, the $-M(i+1)$ in
\ref{match:5}, and the $-M(j+1)$ in \ref{match:6}, in any optimal solution $a_{i}$ is matched to $c_{j}$,
$a_{i-1}$ to $c_{j-1}$, $a_{i-2}$ to $c_{j-2}$ and so on. The total cost of these matching is $w(i,j)$ since the $w(i,j)-w(i-1,j-1)$ terms in \ref{match:4} form a telescoping sum.

\begin{theorem}
\label{thm:triangle}
For sufficiently large $M$, the total cost of an optimal matching in a hard instance with costs
\ref{match:1}-\ref{match:6} is $-3M^{2}+\min_{i,j,k}w(i,k)+w(k,j)+w(i,j)$.
\end{theorem}

\begin{proof}
Consider $i,j,k$ minimizing $w(i,k)+w(k,j)+w(i,j)$. 
We assume without loss of generality that $i\geq j$. It is easy to see that it is possible to choose the
following matching (see Figure~\ref{fig:macrooptimal}):
\begin{enumerate}
\item $b'_{k}$ to $d'_{k}$ with cost $-M^{2}-2M\cdot k$.
\item $b_{k+1}$ to $c'_{j}$ with cost $-M^{2}+M\cdot k+M\cdot j+w(k,j)$.
\item  $d_{k+1}$ to $a'_{i}$ with cost $-M^{2}+M\cdot k+M\cdot i+w(i,k)$.
\item $a_{i}$ to $c_{j}$, $a_{i-1}$ to $c_{j-1}$, $a_{i-2}$ to $c_{j-2}$, \ldots, $a_{i-j+2}$ to $c_{2}$
with costs $-2M+w(i,j)-w(i-1,j-1),-2M+w(i-1,j-1)-w(i-2,j-2),\ldots,-2M+w(i-j+2,2)-w(i-j+1,1)$. 
\item $a_{i-j+1}$ to $c_{1}$ with cost $-M \cdot (i-j+2)+w(i-j+1,1)$.
\end{enumerate}
Summing up and telescoping, the total cost is
$
-M^{2}-2M\cdot k -M^{2}+M\cdot k+M\cdot j+w(k,j)-M^{2}+M\cdot k+M\cdot i+w(i,k)-2M\cdot (j-1)+w(i,j)-M\cdot (i-j+2)
$
which is equal to $-3M^{2}+w(i,k)+w(k,j)+w(i,j)$.

For the other direction,
we need to prove that every solution has cost at least $-3M^{2}+\min_{i,j,k} w(i,k)+w(k,j)+w(i,j)$. We first observe that, by Lemma~\ref{lem:redstructure}, a solution can match
$b'_{k}$ to $d'_{k}$ at most once for some $k$. Similarly, it can match $b_{k'+1}$ to $c'_{j}$ at most once
for some $j$ and $k'$, and $d_{k''+1}$ to $a'_{i}$  at most once for some $i$ and $k''$.
Furthermore, for $M$ large enough, either the cost is larger than $-3M^{2}$ or all three such
pairs of nodes are matched for some $k,i,j,$ and $k',k'' \geq k$. Furthermore, if $k'>k$ and $M$ is large enough
then we can decrease $k'$ by one thus decreasing the total cost, and similarly if $k''>k$. 
It is enough to consider an optimal solution and hence we can assume that $k=k'=k''$. 

Again by Lemma~\ref{lem:redstructure},
the only possible additional matched pairs of nodes are a subsequence of spine nodes $a_{1},\ldots,a_{i}$
and $c_{1},\ldots,c_{j}$. We show that an optimal solution matches $a_i$ with $c_j$, $a_{i-1}$ with $c_{j-1}$, ..., $a_{i-j+1}$ with $c_1$.
To this end, suppose that $a_{x_{z}}$ is matched to $c_{y_{z}}$, for every $z=1,2,\ldots,L$, where $1\leq x_{1}<\ldots <x_{L}\leq i$
and $1\leq y_{1}<\ldots <y_{L}\leq j$.
For every $z$ this contributes, up to lower order terms less than $M$, 
$-2M$ if $x_{z},y_{z}>1$, or $-M(y_{z}+1)$ if $x_{z}=1$, or $-M(x_{z}+1)$ if $y_{z}=1$.
In an optimal solution we have $x_{1}=1$ or $y_{1}=1$, as otherwise we can match $a_{1}$ to $c_{1}$ to decrease
the total cost.
First, assume that $y_{1}=1$. Then, if $x_{z'}+1<x_{z'+1}$ for some $z'\in\{1,\ldots,L\}$ (where we define
$x_{L+1}=i+1$), we can increase all $x_{1},x_{2},\ldots,x_{z'}$ by 1 to decrease the total cost by $M$, up to
lower order terms. So $x_{L}=i, x_{L-1}=i-1,\ldots, x_{1} = i-L+1$. Now if $L<j$ then $x_{1}>1$ (recall that we assumed $i \ge j$) and also
$y_{z'}+1<y_{z'+1}$ for some $z'\in\{1,\ldots,L\}$ (again, we define $y_{L+1}=j+1$). This means
that we can increase all $y_{1},y_{2},\ldots,y_{z'}$ by 1 and then additionally match $a_{x_{1}-1}$ with
$c_{1}$ to decrease the total cost by $M$, up to lower order terms. 
Second, if $x_{1}=1$ a symmetric argument applies.
We obtain that indeed $L=\min(i,j)=j$, and
$a_{i-j+1}$ is matched to $c_{1}$, $a_{i-j+2}$ is matched to $c_{2}$, ..., $a_i$ is matched to $c_j$.
Now, by the same calculations as in the previous
paragraph, the total cost is $-3M^{2}+w(i,k)+w(k,j)+w(i,j)$.
\end{proof}

\section{Reducing Max-Weight \boldmath$k$-Clique to TED}
\label{sec:kclique}

The drawback of the reduction described in Section~\ref{sec:apsp} is the large size of the alphabet.
That is, given a complete weighted $n$-node undirected graph it creates two trees of size $O(n)$
where labels of nodes are distinct, and therefore $|\Sigma|=\Theta(n)$. We would like to refine
the reduction so that $|\Sigma|=O(1)$. However, as the input size of TED on $n$-node trees and alphabet $\Sigma$ with $O(\log n)$-bit integer weights is $\tilde O(n + |\Sigma|^2)$, such a reduction would need to compress the $\tilde O(n^2)$ input size of negative triangle detection considerably. 
To circumvent this barrier, we assume the stronger Max-Weight $k$-Clique
Conjecture, where the input size $\tilde O(n^2)$ is very small compared to the running time bound $O(n^k)$. 

\begin{lemma}
\label{lem:kcliquereduction}
Given a complete undirected $n$-node graph $G$ with weights
in $\{1,\ldots,n^{ck}\}$, we construct, in linear time in the output size, an instance of TED of size $O(n^{k/3+2})$ with alphabet size $|\Sigma|=O(ck)$
such that the maximum weight of an $k$-clique in $G$ can be extracted from the edit distance. 
\end{lemma}

Thus,
an $O(n^{3-\epsilon'})$ time algorithm for TED for sufficiently large $|\Sigma|=O(1)$ implies an $O(n^{(k/3+2)(3-\epsilon')})$
time algorithm for max-weight $k$-Clique. Setting $\epsilon=\epsilon'/6$, we obtain that, for every $c>0$,
there exists $k=\lceil 6/\epsilon \rceil$ such that max-weight $k$-Clique can be solved in time
\[
O(n^{(k/3+2)(3-\epsilon')}) = O(n^{k-\epsilon'k/3+6-2\epsilon'}) = O(n^{k(1-\epsilon)-k\epsilon+6}) = O(n^{k(1-\epsilon)}),
\]
so Conjecture~\ref{conj:maxclique} is violated.

The reduction starts with enumerating all $\frac{k}{3}$-cliques in the graph and identifying them with
numbers $1,2,\ldots,N$, where $N\leq n^{k/3}$. Let $Q(i)$ denote the set of nodes in the $i$-th clique.
Then, for $i,j$ such that $Q(i)\cap Q(j)=\emptyset$,
$W(i,j)$ is the total weight of all edges connecting two nodes in the $i$-th clique or a node in the $i$-th
clique with a node in the $j$-th clique.
Our goal is to calculate the maximum value of $W(i,z)+W(z,j)+W(j,i)$ over $i,j,z$ such that $Q(i),Q(j)$ and $Q(z)$
are pairwise disjoint.
If we define $w(u,u)=0$ and increase every other weight $w(u,v)$ by $\Lambda := k^{2} n^{ck}$, this is equivalent to
maximising over all $i,j,z$. Indeed, if $Q(i),Q(j),Q(z)$ are pairwise disjoint, the total weight is
at least ${k \choose 2} \Lambda$, and otherwise it is at most $\big({k \choose 2} - 1\big) (\Lambda + n^{ck}) < {k \choose 2} \Lambda$.
Note that the new weights are still bounded by $n^{O(ck)}$.

Our construction of a hard instance of size $O(N \cdot \textup{poly}(n))$ is similar to Section~\ref{sec:apsp}, however, 
the costs are set up differently and we attach small additional gadgets to some of the nodes (which is necessary, cf.\ Section~\ref{sec:algo}).
The original two trees (with some extra spine nodes without any leaves) are called the \emph{macro structure}
and all small gadgets are called the \emph{micro structures}. With notation as in Section~\ref{sec:apsp}, the following micro structures are created
for every $i=1,2,\ldots,N$ (see Figure~\ref{fig:macro2}):
\begin{enumerate}
\item $A'_{i}$ attached to the leaf $a'_{i}$,
\item a copy of $I$ attached as the left child of the leaf $c'_{i}$,
\item $C'_{i}$ attached as the right child of the leaf $c'_{i}$,
\item $A_{i}$ attached to the spine node $a_{i-1}$ between the previously existing children $a_{i}$ and $a'_{i-1}$,
\item $B_{i}$ attached to the spine node $b_{i}$ between the previously existing children $b_{i+1}$ and $b'_{i}$,
\item $C_{i}$ attached to the spine node $c_{i-1}$ as the rightmost child,
\item $D_{i}$ attached to the spine node $d_{i}$ between the previously existing children $d'_{i}$ and $d_{i+1}$.
\end{enumerate}
Notice that $A_{i}$ is attached above $a_{i}$ (and similarly $C_{j}$ is attached above $c_{j}$). 
Therefore, we need to create dummy spine nodes $a_{0}$ and $c_{0}$. We also insert an
additional spine node $b''_{i}$ between $b_{i}$ and $b_{i+1}$ and similarly $d''_{i}$ between
$d_{i}$ and $d_{i+1}$, for every $i=1,2,\ldots,N-1$. See Figure~\ref{fig:macro2}.

\begin{figure}[h!]
\begin{center}
\includegraphics[scale=0.8]{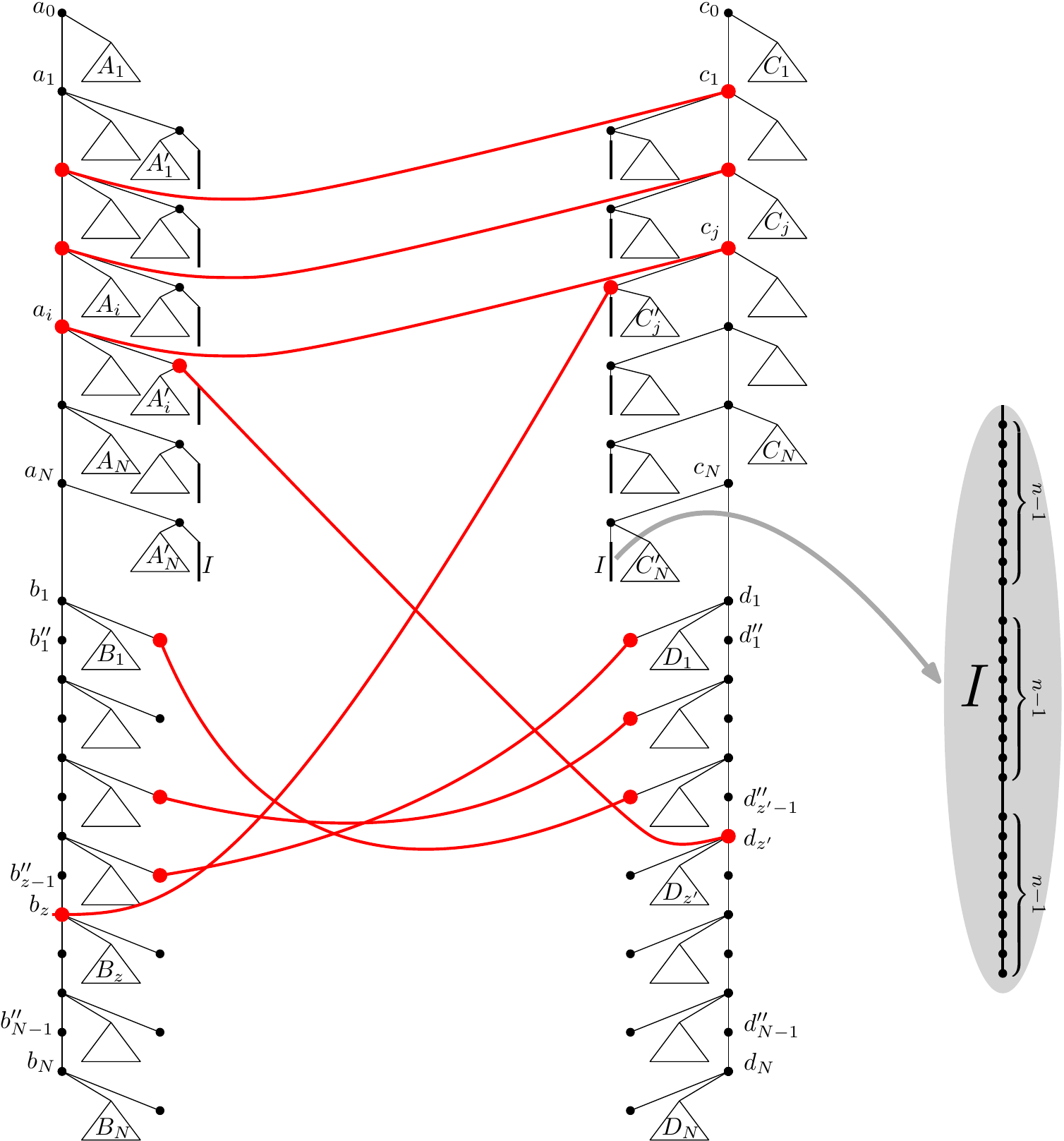}
\caption{A hard instance of TED constructed for a given instance of max-weight $k$-clique. \label{fig:macro2}}
\end{center}
\end{figure}

The costs in the macro structure are chosen as follows, where again $M$ is a sufficiently large value (picking $M=n^{O(ck)}$ will suffice):
\begin{enumerate}
\item $\cmatch(b_{z},c'_{j}) = -M^{8}$ for every $z=1,2,\ldots,N$ and $i=1,2,\ldots,N$,
\item $\cmatch(a'_{i},d_{z'}) = -M^{8}$ for every $i=1,2,\ldots,N$ and $z'=1,2,\ldots,N$,
\item $\cmatch(b'_{z},d'_{z'}) = -M^{7}\cdot 2$ for every $z=1,2,\ldots,N$ and $z'=1,2,\ldots,N$,
\item $\cmatch(a_{i},c_{j}) = -M^{3}\cdot 2+M^{2}$ for every $i=1,2,\ldots,N$ and $j=1,2,\ldots,N$.
\end{enumerate}
Additionally, the extra spine nodes $b''_{i}$ and $d''_{i}$ can be matched to some of the nodes
of $I$. Each copy of $I$ is a path consisting of $k/3$ segments $I_{0},I_{1},\ldots,I_{k/3-1}$
of length $n-1$, where the root of the whole $I$ belongs to $I_{0}$.
The label of every $u\in I_{i}$ is the same and the costs are set so that $\cmatch(u,u)=-M^{7}\cdot n^{i}$.
The label of every $b''_{z}$ (and also $d''_{z'}$) is chosen as the label of every $u\in I_{m}$, where $n^{m}$ is the largest power of $n$ dividing $N-z$.  
The cost of matching any other two labels used in the macro structure is set to infinity.
For the nodes belonging to the other micro structures, the cost of matching is at least $-M^{6}$
and will be specified precisely later. This is enough to enforce the following structural property.

\begin{lemma}
\label{lem:optimalmacro}
For sufficiently large $M$, any optimal matching has the following structure:
there exist $i,j,z$ such that $a'_{i}$ is matched to $d_{z}$, $c'_{j}$ is matched to $b_{z}$,
$b'_{1}$ is matched to $d'_{z-1}$, $b'_{2}$ is matched to $d'_{z-2}$, \ldots, $b'_{z-1}$ is
matched to $d'_{1}$. Furthermore, if $z<N$ then $b''_{z}$ is matched to a descendant
of $c'_{j}$ and $d''_{z}$ is matched to a descendant of $a'_{i}$.
Ignoring the spine nodes
$a_{1},\ldots,a_{i},c_{1},\ldots,c_{i}$ and all micro structures that are not copies of $I$ the cost of any such solution is $-M^{8}\cdot 2-M^{7}\cdot 2(N-1)$.
\end{lemma}

\begin{proof}
For sufficiently large $M$, any optimal solution must match $a'_{i}$ to $d_{z}$ and $c'_{j}$ to $b_{z'}$,
for some $i,j,z,z'$, as otherwise its cost is larger than $-M^{8}\cdot 2$. By the reasoning described in 
Lemma~\ref{lem:redstructure}, these $i,j,z,z'$ are uniquely defined for any optimal solution.

Nodes from the copy of $I$ attached as the left child of the leaf $c'_{j}$ can be matched
to some spine nodes below $b_{z}$, nodes from the copy of $I$ attached as the right
child of the leaf $a'_{i}$ can be matched to some spine nodes below $d_{z'}$, and no other
nodes from the copies of $I$ can be matched. We claim that the total contribution of these
nodes is $-M^{7}(N-z)$ and $-M^{7}(N-z')$, respectively. By symmetry, it is enough to justify
the former. Observe that the cost of matching a single $u\in I_{i}$ is smaller than the total cost of matching
all nodes from $I_{0}\cup\ldots I_{i-1}$, therefore an optimal solution must match as many
nodes to nodes from $I_{k/3-1}$ as possible. Looking at the expansions of all numbers
$N-z,N-(z+1),\ldots,N-(N-1)$ in base $n$, where $N-z=\sum_{i=0}^{k/3-1}\alpha_{i}n^{i}$,
we see that there are $\alpha_{k/3-1}$ such nodes, namely the nodes $b'_{N-z'}$ with $z \le z' < N$ and $N-z'$ divisible by $n^{k/3-1}$. Then, an optimal solution must
match as many nodes to nodes from $I_{k/3-2}$ as possible to nodes above the topmost node
matched to a node from $I_{k/3-1}$. Looking again at the same expression, we see that
there are $\alpha_{k/3-2}$ such nodes, namely the nodes $b'_{N-z'}$ with $z \le z' < N - \alpha_{k/3-1} n^{k/3-1}$ and $N-z'$ divisible by $n^{k/3-2}$. Continuing in the same fashion, we obtain that
there are $\alpha_{i}$ nodes matched to nodes from $I_{i}$, making the total cost $-M^{7}(N-z)$
as claimed.

We assume without loss of generality that $z\geq z'$. Then, an optimal solution must match
$d'_{z'-1}$ to $b'_{x_{z'-1}}$, $d'_{z'-2}$ to $b'_{x_{z'-2}}$, \ldots, and $d'_{1}$ to $b'_{x_{1}}$,
for some $z \geq x_{1} > \ldots > x_{z'-1} \geq 1$, as otherwise its cost is larger than
$-M^{8}\cdot 2-M^{7}(2N-z-z')-M^{7}\cdot 2(z'-1)$. Rewriting the cost we obtain
$-M^{8}\cdot 2-M^{7}(2N-2-z+z')$, so recalling our assumption $z \ge z'$ we see that in fact $z=z'$ as
otherwise its cost is larger than $-M^{8}\cdot 2-M^{7}\cdot 2(N-1)$.
\end{proof}

We are now ready to state properties of the remaining micro structures.
Let $\cmatch(T_{1},T_{2})$ denote the cost of matching two trees $T_{1}$ and $T_{2}$. Then, we require that:
\begin{enumerate}
\item $\cmatch(A'_{i},D_{z'}) = -M^{6} - M^{3}(N-i) - W(i,z')$ for every $i=1,2,\ldots,N$ and $z'=1,2,\ldots,N$,
\item $\cmatch(B_{z},C'_{j}) = -M^{6} - M^{3}(N-j) - W(z,j)$ for every $z=1,2,\ldots,N$ and $j=1,2,\ldots,N$.
\item $\cmatch(A_{i},C_{j}) = -M^{2} - W(j,i) + W(j-1,i-1)$ for every $i=2,3,\ldots,N$ and $j=2,3,\ldots,N$.
\item $\cmatch(A_{i},C_{1}) = -M^{5}-M^{3}(i-1) - W(1,i)$ for every $i=1,2,\ldots,N$,
\item $\cmatch(A_{1},C_{j}) = -M^{5}-M^{3}(j-1) - W(j,1)$ for every $j=1,2,\ldots,N$.
\end{enumerate}
The labels of the nodes in the micro structures should be partitioned into disjoint subsets corresponding
to the following micro structures:
\begin{enumerate}
\item $\{A'_{1},A'_{2},\ldots,A'_{N},D_{1},D_{2},\ldots,D_{N}\}$,
\item $\{B_{1},B_{2},\ldots,B_{N},C'_{1},C'_{2},\ldots,C'_{N}\}$,
\item $\{A_{1},A_{2},\ldots,A_{N},C_{1},C_{2},\ldots,C_{N}\}$,
\end{enumerate}
so that two nodes can be matched only if their labels belong to the same subset.
The cost of matching any node of $A'_{i},D_{z'},B_{z},C'_{j}$ should be at least $-M^{6}$.
The cost of matching any node of $A_{i},C_{j}$ should be at least $-M^{2}$, except
that the root of $A_{i}$ ($C_{j}$) can be matched to the root of $C_{1}$ ($A_{1}$) with cost 
larger than $-M^{5}-M$ but at most $-M^{5}$,
and, for any non-root node of $A_i$ ($C_{j}$) and for any non-root node of
$C_{1}$ ($A_{1}$), the cost of matching is larger than $-M^{4}$.
Finally, every $A_{i}$ and $C_{j}$ should consist of $O(n^{2})$ nodes.
Now we can show that, assuming these properties, any optimal solution has a specific structure.

\begin{lemma}
\label{lem:optimal}
For sufficiently large $M$, the total cost of an optimal matching is
\[
-M^{8}\cdot 2 - M^{7}\cdot 2(N-1) - M^{6}\cdot 2 - M^{5} - M^{3}\cdot 2N + M^{2} - \max_{i,j,z}\{W(i,z) + W(z,j) + W(j,i)\}.
\]
\end{lemma}

\begin{proof}
Consider $i,j,z$ maximizing $W(i,z) + W(z,j) + W(j,i)$.
We may assume that $i\geq j$. Then, it is possible to choose the following matching:
\begin{enumerate}
\item $b_{k}$ to $c'_{j}$ with cost $-M^{8}$,
\item some nodes from the copy of $I$ being the left child of $c'_{j}$ to some spine nodes
below $b_{z}$ with total cost $-M^{7}(N-z)$,
\item $a'_{i}$ to $d_{k}$ with cost $-M^{8}$,
\item some nodes from the copy of $I$ being the right child of $a'_{i}$ to some spine nodes
below $d_{z}$ with total cost $-M^{7}(N-z)$,
\item $b'_{1}$ to $d'_{z-1}$, $b'_{2}$ to $d'_{z-2}$, \ldots, $b'_{z-1}$ to $d'_{1}$ with cost $-M^{7}\cdot 2$ each,
\item $a_{i}$ to $c_{j}$, $a_{i-1}$ to $c_{j-1}$, \ldots, $a_{i-j+1}$ to $c_{1}$ with cost $-M^{3}\cdot 2+M^{2}$ each,
\item $A'_{i}$ to $D_{z}$ with cost $-M^{6}-M^{3}(N-i)-W(i,z)$,
\item $B_{z}$ to $C'_{j}$ with cost $-M^{6}-M^{3}(N-j)-W(z,j)$,
\item $A_{i}$ to $C_{j}$, $A_{i-1}$ to $C_{j-1}$, \ldots, $A_{i-j+2}$ to $C_{2}$ with costs
$-M^{2}-W(j,i)+W(j-1,i-1)$, $-M^{2}-W(j-1,i-1)+W(j-2,i-2)$, \ldots, $-M^{2}-W(2,i-j+2)+W(1,i-j+1)$.
\item $A_{i-j+1}$ to $C_{1}$ with cost $-M^{5}-M^{3}(i-j)-W(1,i-j+1)$.
\end{enumerate}
Summing up and telescoping, the total cost is
\begin{align*}
-&M^{8}\\
-&M^{7}(N-z)\\
-&M^{8}\\
-&M^{7}(N-z)\\
-&M^{7}\cdot 2(z-1)\\
-&M^{3}\cdot 2j+M^{2}\cdot j\\
-&M^{6}-M^{3}(N-i)-W(i-z)\\
-&M^{6}-M^{3}(N-j)-W(z,j)\\
-&M^{2}(j-1)-W(j,i)-M^{5}-M^{3}(i-j)\\
=& -M^{8}\cdot 2 - M^{7}\cdot 2(N-1) - M^{6}\cdot 2 - M^{5} - M^{3}\cdot 2N + M^{2} - W(i,z) - W(z,j) - W(j,i).
\end{align*}

For the other direction, we need to argue that every solution has cost at least
$-M^{8}\cdot 2 - M^{7}\cdot 2(N-1) - M^{6}\cdot 2 - M^{5} - M^{3}\cdot 2N + M^{2} - \max_{i,j,z}\{W(i,z) + W(z,j) + W(j,i)\}$. We start with invoking Lemma~\ref{lem:optimalmacro} and analyse the remaining
small micro structures. Due to leaves $b'_{1},\ldots,b'_{z-1},d'_{1},\ldots,d'_{z-1}$ being already
matched, no node from $B_{1},\ldots,B_{z-1},D_{1},\ldots,D_{z-1}$ can be matched (as they can in general only be matched to $A'_*$'s and $C'_*$'s). Then,
due to $b''_{z}$ and $d''_{z}$ being already matched (or $z=N$) no node from
$B_{z+1},\ldots,B_{N},D_{z+1},\ldots,D_{N}$ can be matched, and nodes from $B_{z}$ or $D_{z}$
can be only matched to nodes from $C'_{j}$ or $A'_{i}$, respectively. The cost incurred by all such
nodes is $\cmatch(A'_{i},D_{z})+\cmatch(B_{z},C'_{j})$, making the total cost
$-M^{8}\cdot 2-M^{7}\cdot 2(N-1)-M^{6}\cdot 2-M^{3}(2N-i-j)-W(i,z)-W(z,j)$.
It remains to analyse the contribution of all spine nodes $a_{1},\ldots,a_{N},b_{1},\ldots,b_{N}$
and nodes from micro structures $A_{1},\ldots,A_{N},C_{1},\ldots,C_{N}$.

Consider the micro structures $C_{1}$ and $A_{1}$. Matching their roots to roots of some
$A_{i'}$ and $C_{j'}$, respectively, decreases the total cost by at least $-M^{5}$, which is much smaller than
the cost of matching the remaining nodes. Furthermore, it is not possible to match
both the root of $C_{1}$ to the root of some $A_{i'}$ and the root of $A_{1}$ to the root
of some $C_{j'}$ at the same time, unless the root of $A_{1}$ is matched to the root of $C_{1}$.
Therefore, an optimal solution matches exactly one of them or both to each other, say we match the root of $C_{1}$ to the root
of some $A_{i'}$, thus adding $\cmatch(A_{i'},C_{1})$ to the total cost. 
Due to $a'_{i}$ being matched to $d_{z}$, $i'\leq i$ holds. Now, unless $i'=1$,
no node from $A_{1}$ can be matched to a node from $C_{j'}$, so the cost of matching any
$a_{i'}$ to $c_{j'}$ is now much smaller than the cost of matching nodes in the remaining
micro structures (for each such node, the cost is at least $-M^{2}$, and there are at most
$O(n^{2})$ of them in a single micro structure, so the total cost contributed by a single micro structure
is larger than $-M^{3}$ for $M$ large enough) and, by Lemma~\ref{lem:redstructure},
only nodes $a_{1},\ldots,a_{i},c_{1},\ldots,c_{j}$ can be matched, so an optimal solution matches 
as many such pairs as possible. Due to the root of $C_{1}$ being matched to the
root of $A_{i'}$, only nodes $a_{i'},a_{i'+1},\ldots,a_{i}$ and $c_{1},\ldots,c_{j}$ can be
matched, so there are $\min(i-i'+1,j)$ such matched pairs.
If $i-i'+1<j$ and $i'>1$ then $C_{1}$ can be matched with
$A_{i'-1}$ instead of $A_{i'}$ which allows for an additional pair and decreases the total
cost (because matching a pair $(a_*,c_*)$ adds $-M^{3}\cdot 2$ to the cost while decreasing $i'$ by one adds
$M^{3}$ to the cost $\cmatch(A_{i'},C_1)$, up to lower order terms). If $i-i'+1<j$ and $i'=1$ then $A_{1}$ can be matched with
$C_{2}$ instead of $C_{1}$ while keeping the number of matched pairs intact to decrease the total cost.
So $i-i'+1\geq j$ (implying $i\geq j$, which is due to our initial assumption that the root of $C_{1}$
is matched to the root of some $A_{i'}$). Then, if $i'<i-j+1$, $C_{1}$ can be matched with
$A_{i'+1}$ instead of $A_{i'}$ without changing the number of matched pairs to decrease the
total cost. Thus, $i'=i-j+1$ and $a_{i}$ is matched to $c_{j}$, $a_{i-1}$ to $c_{j-1}$, \ldots,
and $a_{i-j+1}$ to $c_{1}$, Then nodes from $A_{i}$ can be only matched to nodes from $C_{j}$,
nodes from $A_{i-1}$ only to nodes from $C_{j-1}$, and so on. By the same calculations as in the previous
paragraph, the total cost is therefore $-M^{8}\cdot 2 - M^{7}\cdot 2(N-1) - M^{6}\cdot 2 - M^{5} - M^{3}\cdot 2N + M^{2} - \max_{i,j,z} \{W(i,z) + W(z,j) + W(j,i)\}$.
\end{proof}

To complete the proof we need to design the remaining micro structures.
We start with describing some preliminary gadgets that will be later appropriately composed to
obtain the micro structures $A_{i},A'_{i},B_{z},C_{j},C'_{j},D_{z'}$ with the required properties. Each
such gadget consists of two trees, called left and right, and we are able to exactly calculate the cost of matching
them. The main difficulty here is that we need to keep the size of the alphabet small, so for instance
we are not able to create a distinct label for every node of the original graph. At this point it is
also important to note that we can assume $M=n^{O(ck)}$, i.e., there is a constant $d = O(ck)$ such that all weights constructed above have absolute value less than $n^d$.

\paragraph{Decrease gadget $D(x)$.} \label{decreasegadget} For any $x\in \{0,\ldots,n^d-1\}$, the edit distance of the left and right tree of $D(x)$ is
$-x$, and furthermore the right tree does not depend on the value of $x$. 

This is obtained
by representing $x$ in base $n$ as $x=\sum_{i=0}^{d-1}\alpha_{i}n^{i}$. The left tree is a
path composed of $d$ segments, the $i$-th segment consisting of $\alpha_{i}$ nodes. The right
tree is a path composed of $d$ segments, each consisting of $n-1$ nodes. Nodes from the
$i$-th segment of the left tree can be matched with nodes from the $i$-th segment of the right
tree with cost $-n^{i}$, so the total cost is $-x$, see Figure~\ref{fig:decreaseequality} (left).
We reuse the same set of distinct labels in
every decrease gadget of the same type, hence we need only $O(d)$ distinct labels in total.
Furthermore, note that the cost of matching the left tree of $D(x)$ with any tree is at least $-x$
and the cost of matching any node of $D(x)$ is $-n^{i}$ for some $i\in\{0,1,\ldots,d-1\}$.

\begin{figure*}[h!]
    \centering
    \begin{subfigure}{0.5\textwidth}
       \centering
\includegraphics[scale=0.8]{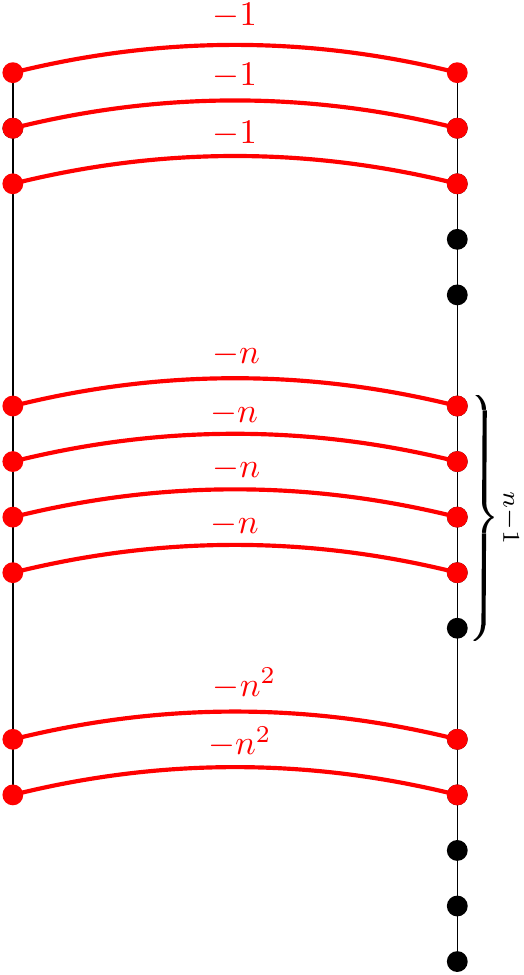}
    \end{subfigure}
    \begin{subfigure}{0.5\textwidth}
        \centering
\includegraphics[scale=0.8]{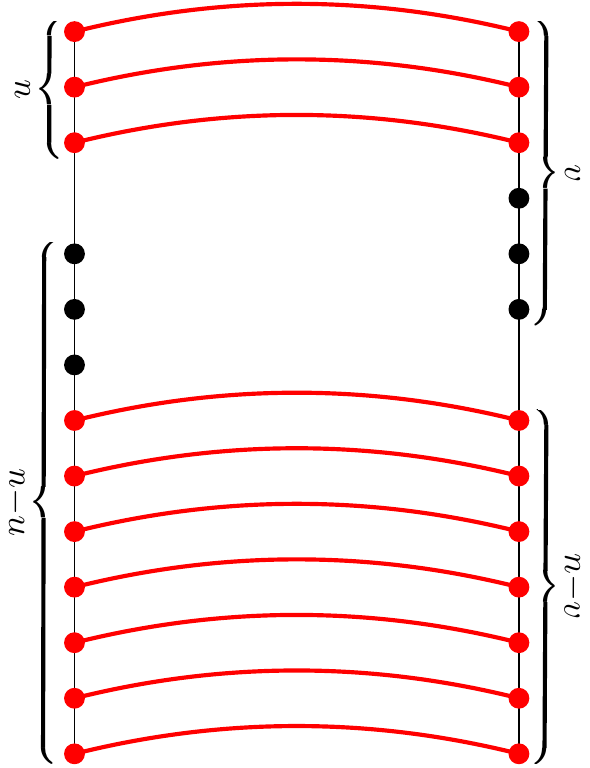}
    \end{subfigure}
\caption{\label{fig:decreaseequality} Left: Decrease gadget built for $d=3$, $n=6$ and $x=n^{2}\cdot 2+n\cdot 4+3$. Right: Equality gadget for $u=3$, $v=6$.}

\end{figure*}

\paragraph{Equality gadget $E(u,v,c_{=})$.} For any $u,v\in \{1,\ldots,n\}$
and $c_{=}\in \{0,\ldots,n^{d}-1\}$, the edit distance of the left and right tree of $E(u,v,c_{=})$ is
$-c_{=}\cdot n$ if $u=v$ and at least $-c_{=}\cdot n+c_{=}$ otherwise. Also, the left tree does not
depend on $v$ and the right tree does not depend on $u$.

The left tree is a path composed of a segment of length $u$ and a segment of length $n-u$. The right
tree is a path composed of a segment of length $v$ and a segment of length $n-v$. Nodes from the first
segment of the left tree can be matched with nodes from the first segment of the right tree with cost $-c_{=}$,
and similarly for the second segments. Then, if $u=v$ we can match all nodes in both trees,
so the total cost is $-c_{=}\cdot n$. Otherwise, we can match at most $n-1$ nodes, making
the total cost at least $-c_{=}\cdot n+c_{=}$, see Figure~\ref{fig:decreaseequality} (right). Furthermore, note that
the total cost of matching the left tree of $E(u,v,c_{=})$ with any tree is at least
$-c_{=}\cdot n$ and the cost of matching any node of $E(u,v,c_{=})$ is $-c_{=}$.

\paragraph{Connection gadget $C(i,j,M)$.} For any $i,j\in\{1,\ldots,N\}$ and sufficiently large
$M\in \{0,\ldots,n^{d}-1\}$, the edit distance of the left and right tree of $C(i,j,M)$ is $-M-W(i,j)$.
The left tree does not depend on $j$ and the right tree does not depend on $i$.

Let $\{u_{1},\ldots,u_{k/3}\}$ and $\{v_{1},\ldots,v_{k/3}\}$
be the $k/3$-cliques corresponding to $i$ and $j$, respectively,
where $u_{1}<u_{2}<\ldots,u_{k/3}$ and $v_{1}<v_{2}<\ldots<v_{k/3}$.
Recall that $W(i,j)$ denotes the total weight of all edges
connecting two nodes in the $i$-th clique or a node in the $i$-th clique with a node in the $j$-th
clique, where we assume that $w(u,u)=0$.
We construct the gadget $C(i,j,M)$ as follows.
The root of the left tree has degree $1+k/3$ and the root of the right tree has degree $1+n$.
Their rightmost children correspond to the root of the left and the right trees of
$D(\sum_{x<y}w(u_{x},u_{y}))$, respectively. Every other child of the left root can be matched
with every other child of the right root with cost $-M_2$ (we fix $M_1$ and $M_2$ later). Intuitively, we would like the
$x$-th child of the the left root to be matched with the $u_{x}$-th child of the right root,
and then contribute $-\sum_{y}w(u_{x},v_{y})$ to the total cost, so that summing up over
$x=1,2,\ldots,k/3$ we obtain the desired sum.
To this end, we attach the left tree of $E(u_{x},\cdot,M_{1})$
and the right tree of $D(\cdot)$ to the $x$-th child of the left root. Similarly, we attach
the right tree of $E(\cdot,t,M_{1})$ and the left tree of $D(\sum_{y}w(t,v_{y}))$ to
the $t$-th child of the right root.
Here we use $\cdot$ to emphasise that a particular tree does not depend on the particular
value of the parameter.
All decrease gadgets are of the same type. See Figure~\ref{fig:connection}.

\begin{figure}[h!]
\begin{center}
\includegraphics[scale=0.8]{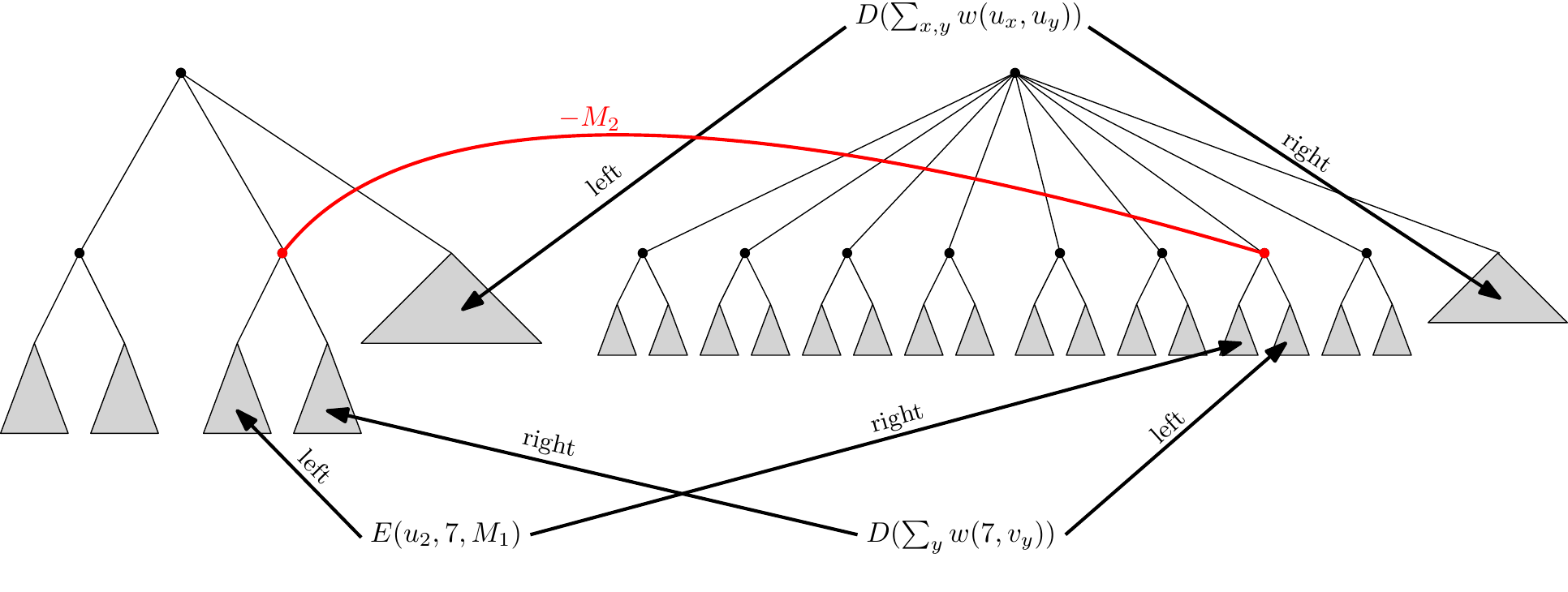}
\caption{Schematic illustration of a connection gadget for $k/3=2$ and $n=8$.\label{fig:connection}}
\end{center}
\end{figure}

We can clearly construct a solution with total cost $-M_{2}\cdot k/3-M_{1}\cdot n\cdot k/3-W(i,j)$
(because we have enumerated the clique corresponding to $i$ so that $u_{1}<u_{2}<\ldots<u_{k/3}$).
We claim that, for appropriately chosen $M_{1}$ and $M_{2}$, no better solution is possible.
Let $W=\sum_{u,v}w(u,v)$.
We fix $M_{1}=W \cdot (k/3+1)$. This is enough to guarantee
that the total cost contributed by nodes in all decrease gadgets is at least $-M_{1}$.
The total cost contributed by nodes in all equality gadgets is at least $-M_{1}\cdot n\cdot k/3$.
Consequently, setting $M_{2}=M_{1}\cdot n\cdot k/3+M_{1}$ guarantees that any optimal
solution must match all children of the left root, so in fact, for every $x=1,2,\ldots,k/3$ we
must match the $x$-th child of the left root to some child of the right root.
Because matching the left tree of any decrease gadget contributes at least $-W$ to the
total cost, by the choice of $M_{1}$ an optimal solution in fact must match
the $x$-th child of the left root with the $u_{x}$-th child of the right root, as otherwise
we lose at least $M_{1}$ due to the corresponding equality gadget that cannot be compensated
by matching its accompanying decrease gadget.
Finally, the corresponding decrease gadget adds $-\sum_{y}w(u_{x},v_{y})$ to
the total cost. Therefore, as long as $M\geq M_{2}\cdot k/3+M_{1}\cdot n\cdot k/3$
the total cost is indeed $-M-W(i,j)$ after choosing the cost of matching the roots to
be $-M+M_{2}\cdot k/3+M_{1}\cdot n\cdot k/3$. 
For any node in a decrease gadget, the cost of matching is at least $-W$, for any node
in an equality gadget, the cost of matching is $-M_{1}$, and finally the cost of matching
the children of the roots is $-M_{2}$, so the cost of matching any node of $C(i,j,W)$
is at least $-M$. For the correctness of the construction it is enough that $M$ is at least
\begin{align*}
M_{2}\cdot k/3+M_{1}\cdot n\cdot k/3 =& M_{1}((nk/3+1)k/3+nk/3)\\
=& W(k/3+1)k/3(nk/3+1+n)\\
=&W(k/3+1)k/3(n(k/3+1)+1)\\
\leq& W\cdot n(k/3+1)^{3} = n^{O(ck)}.
\end{align*} 
 
\paragraph{Micro structures $A'_{i},D_{z'},B_{z},C'_{j}$.} We only explain how to construct
$A'_{i}$ and $D_{z'}$, for any $i=1,2,\ldots,N$ and $z'=1,2,\ldots,N$, as the construction
of $B_{z}$ and $C'_{j}$ is symmetric. Recall that we require $\cmatch(A'_{i},D_{z'})=-M^{6}-M^{4}(N-i)-W(i,z')$
and for every node in $A'_{i}$ and $D_{z'}$ the cost of matching should be at least $-M^{6}$.

$A'_{i}$ consists of a root to which we attach the left tree of $D(M^{6}+M^{4}(N-i)-M)$ and
the left tree of $C(i,\cdot,M)$, while $D_{z'}$ consists of a root to which we attach the right
tree of $D(\cdot)$ and the right tree of $C(\cdot,z',M)$. All decrease gadgets attached
as the left children of $A'_{i}$ and $D_{z'}$ are of the same type, and all decrease gadgets
used inside the connection gadgets attached as the right children are also of the same
but other type. This guarantees that the total cost of matching $A'_{i}$ to $D_{z'}$
is simply $-M^{6}-M^{4}(N-i)+M-M-W(i,j)=-M^{6}-M^{4}(N-i)-W(i,j)$. For sufficiently
large $M\geq W\cdot n(k/3+1)^{3}$, 
the cost of matching any node in $D(M^{6}+M^{4}(N-i)-M)$ is at least $-M^{6}$
and the cost of matching any node in $C(i,j,M)$ is at least $M$.

\paragraph{Micro structures $A_{i},C_{j}$.} Here the situation is a bit more complex,
as we simultaneously require that $\cmatch(A_{i},C_{j})=-M^{2}-W(j,i)+W(j-1,i-1)$
for every $i=2,3,\ldots,N$ and $j=2,3,\ldots,N$ and
$\cmatch(A_{i},C_{1})=-M^{5}-M^{3}(i-1)-W(1,i)$, and $\cmatch(A_{1},C_{j})=-M^{5}-M^{3}(j-1)-W(j,1)$
for every $i=1,2,\ldots,N$ and $j=1,2,\ldots,N$. We must also make sure that
the cost of matching a node of $A_{i}$ to a node of $C_{j}$ should be at least $-M^{2}$, except
that the root of $A_{i}$ ($C_{j}$) can be matched to the root of $C_{1}$ ($A_{1}$) 
with cost larger than $-M^{5}-M$ but at most $-M^{5}$
and, for any non-root node of $A_i$ ($C_{j}$) and for any non-root node of $C_{1}$ ($A_{1}$),
the cost of matching is larger than $-M^{4}$.

For every $i>1$ ($j>1$), $A_{i}$ ($C_{j}$)
consists of two subtrees, called left and right, attached to the common root, while $A_{1}$ ($C_{1}$)
consists of a single subtree connected to a root. For every $i>1$ ($j>1$), the left subtree of $A_{i}$
(the right subtree of $C_{j}$) consists of a root with two subtrees, called left-right and left-right
(right-left and right-right). 
For every $i>1$, nodes of the right subtree of $A_{i}$ can only be matched to nodes of $C_{1}$ and nodes
of the left subtree  of $A_{i}$ can only be matched to nodes of the right subtree of $C_{j}$ for any $j>1$.
For every $j>1$, nodes of the left subtree of $C_{j}$ can only be matched to nodes of $A_{1}$ and nodes
of the right subtree of $C_{j}$  can only be matched to nodes of the left subtree of $A_{i}$ for any $i>1$.
Nodes of $A_{1}$ can be matched to nodes of the left
subtree of $C_{j}$, for any $j>1$.
Nodes of $C_{1}$ can be matched to nodes of the right
subtree of $A_{i}$, for any $i>1$.
Additionally, the root of $A_{1}$ can be matched to the root of $C_{1}$ with cost
$-M^{5}-W(1,1)>-M^{5}-M$, and
for any $i>1$ ($j>1$), the root of $A_{i}$ ($C_{j}$) can be matched to the root of $C_{1}$ ($A_{1}$)
with cost $-M^{5}$. The subtrees are constructed as follows:
\begin{enumerate}
\item the right subtree of $A_{i}$ is the left tree of $D(M^{3}(i-1)+W(1,i))$,
\item the only subtree of $A_{1}$ is the right tree of $D(\cdot)$,
\item the left subtree of $C_{j}$ is the left tree of $D(M^{3}(j-1)+W(j,1))$,
\item the only subtree of $C_{1}$ is the right tree of $D(\cdot)$.
\end{enumerate}
It remains to fully define the left subtree of every $A_{i}$ and the right subtree of
every $C_{j}$, for $i,j>1$.
Recall that the goal is to ensure $\cmatch(A_{i},C_{j})=-M^{2}-W(j,i)+W(j-1,i-1)$.
We define a new $n$-node graph with weight function $w'(u,v)=M-w(u,v)$ for any
$u\neq v$ (for sufficiently large $M$, the new weights are positive). Then, let
$C'(i,j,M)$ denote the connection gadget $C(i,j,M)$ constructed
for the new graph. Nodes of the left-left (left-right) subtree
of $A_{i}$ can be only matched to nodes of the right-left (right-right) subtree of $C_{j}$.
The subtrees are constructed as follows:
\begin{enumerate}
\item the left-left subtree of $A_{i}$ is the left tree of $C(i,\cdot,M\cdot (k/3)^{2})$,
\item the right-left subtree of $C_{j}$ is the right tree of $C(\cdot,j,M\cdot (k/3)^{2})$,
\item the left-right subtree of $A_{i}$ is the left tree of $C'(i-1,\cdot,M^{2})$,
\item the right-right subtree of $C_{j}$ is the right tree of $C'(\cdot,j-1,M^{2})$.
\end{enumerate}
See Figure~\ref{fig:microac}. For the construction of $C(i,j,M\cdot (k/3)^{2})$ and
$C(i-1,j-1,M^{2})$ to be correct we need that $M \cdot (k/3)^{2}\geq W\cdot n(k/3+1)^{3}$
and $M^{2} \geq M\cdot n^{2} \cdot n(k/3+1)^{3}$, respectively, which holds for sufficiently
large $M$.

Now we calculate $\cmatch(A_{i},C_{j})$. $C(i-1,j-1,M^{2})$ contributes $-M^{2}$ minus the total cost of edges
connecting two subsets of $k/3$ nodes in the new graph. As the weights in the new
graph are defined as $w'(u,v)=M-w(u,v)$, this is exactly $-M^{2}-(M\cdot (k/3)^{2}-W(i-1,j-1))$.
$C(i,j,M\cdot (k/3)^{2})$ contributes $-M\cdot (k/3)^{2}-W(i,j)$, so
$\cmatch(A_{i},C_{j})=-M\cdot (k/3)^{2}-W(i,j)-M^{2}-(M\cdot (k/3)^{2}-W(i-1,j-1))=-M^{2}-W(i,j)+W(i-1,j-1)$
as required.

It remains to bound the cost of matching nodes.
Nodes in the left subtree of $A_{i}$ ($C_{j}$) can be matched only to nodes of $C_{1}$
($A_{1}$) with cost at least $-M^{3}\cdot n > -M^{4}$, except that the roots can be
matched with cost $-M^{5}$. The cost of matching a node of $A_{i}$ to a node of $C_{j}$,
for $i,j>1$, is either at least $-M\cdot (k/3)^{2}$ (for the nodes of $C(i,j,M\cdot (k/3)^{2}$)
or at least $-M^{2}$ (for the nodes of $C(i,j,M^{2}))$, so for sufficiently large $M$ at least $-M^{2}$.

\begin{figure}[h!]
\begin{center}
\includegraphics[scale=0.8]{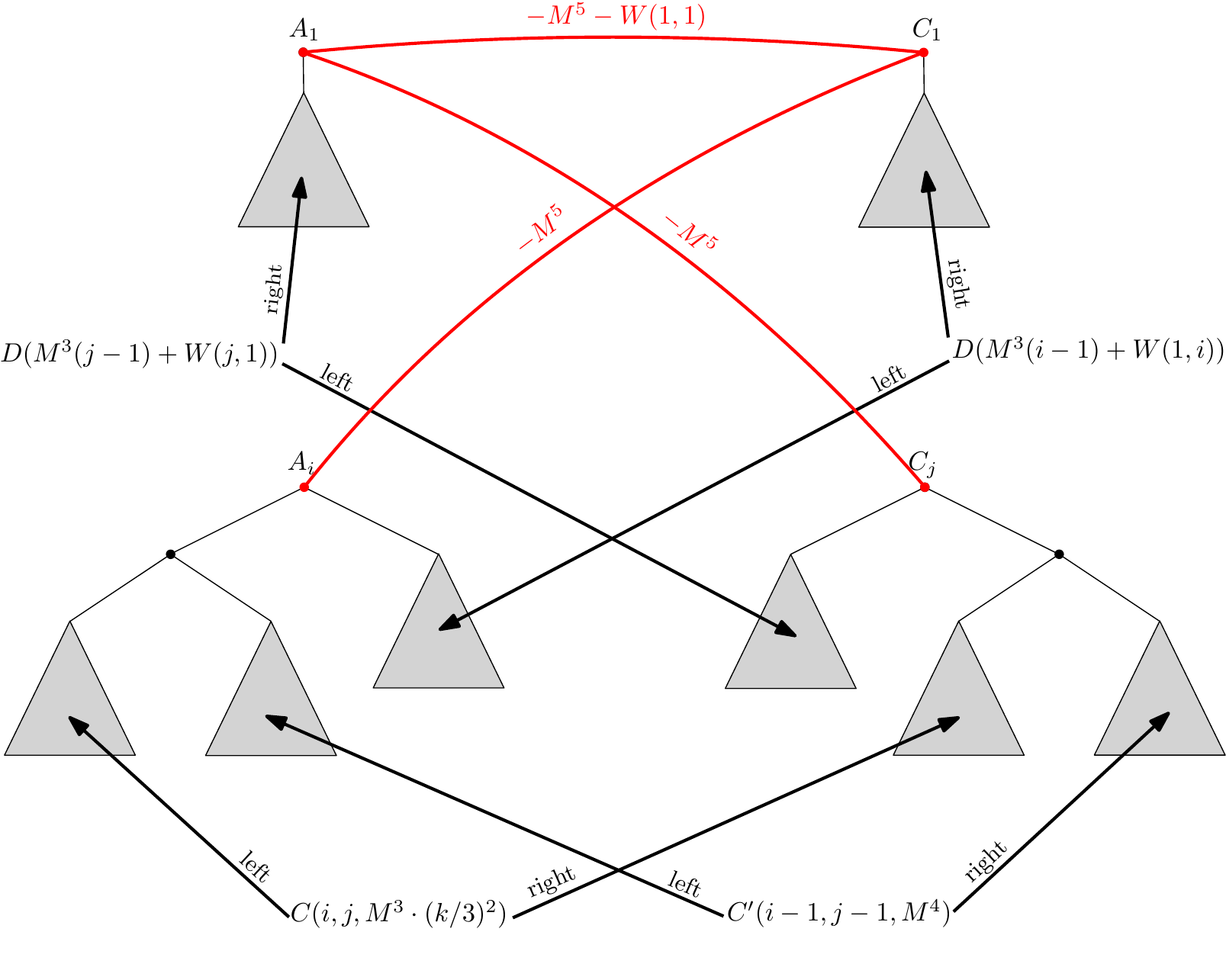}
\caption{\label{fig:microac} Micro structures $A_{1},C_{1}$ and $A_{i},C_{j}$ for $i,j>1$.}
\end{center}
\end{figure}

\paragraph{Wrapping up.} We have shown how to construct, given a complete undirected $n$-node
graph $G$, two trees such that the weight of the max-weight $k$-clique in $G$ can be extracted from
the cost of an optimal matching (and, as mentioned in the beginning of Section~\ref{sec:apsp}, by
a simple transformation this is equal to the edit distance).
To complete the proof of Lemma~\ref{lem:kcliquereduction}, we need
to bound the size of both trees and also the size of the alphabet used to label their nodes.

Initially, each tree consists of $2N$ original spine nodes, where $N\leq n^{k/3}$, $2N$ leaf nodes,
and $N$ additional spine nodes. Then, we attach appropriate microstructures to the original
spine nodes and leaf nodes. The microstructures are $A'_{i},I,C'_{j},A_{i},B_{z},C_{j},D_{z'}$. 
Every copy of $I$ consists of $k/3\cdot n$ nodes.
To analyse the size of the remaining microstructures, first note that if $x\in\{0,\ldots,n^{d}\}$ then the decrease
gadget $D(x)$ consists of $O(d\cdot n)$ nodes. The equality gadget always consists of $O(n)$ nodes.
Finally, the connection gadget $E(\cdot,\cdot,M)$ with $M\in\{0,\ldots,n^{d}\}$ consists of
$O(n(n+d\cdot n)+d\cdot n)=O(d\cdot n^{2})$ nodes. 
Let $M=n^{d}$ with $d$ to be specified later. Now, the size of the microstructures can be
bounded as follows:
$A'_{i}$ and $D_{z'}$ (and also $B_{z}$ and $C'_{j}$) consist of $O(6d\cdot n+d\cdot n^{2})=O(d\cdot n^{2})$ nodes.
The right subtree of $A_{i}$ (and the left subtree of $C_{j}$) consists of $O(3d\cdot n)$ nodes,
while the left subtree of $A_{i}$ (and the right subtree of $C_{j}$) consist of
$O(k^{2}d\cdot n^{2}+2d\cdot n^{2})=O(k^{2}d\cdot n^{2})$ nodes.
Thus, the total size of all microstructures is $O(N\cdot k^{2}d\cdot n^{2})$. It remains to bound $M$.
Recall that we require $M\geq W\cdot n(k/3+1)^{3}$, $M\cdot (k/3)^{2}\geq W\cdot n(k/3+1)^{3}$
and $M \geq n^{3}(k/3+1)^{3}$, where $W \leq n^{2}\cdot n^{O(ck)}=n^{O(ck)}$. Hence, it is sufficient
that $M \geq 8 W \cdot n^{3}k^{3}$.
Setting $d=\Theta(ck)$ is therefore enough. The size of the whole instance thus is
$O(n^{k/3+2}\cdot ck)=O(n^{k/3+2})$.

We also have to bound the size of the alphabet. We need $k/3$ distinct labels for the nodes of~$I$.
We need $O(d)$ distinct labels for the nodes of all decrease gadgets of the same type. There is a constant
number of types, and all other nodes require only a constant number of distinct labels
(irrespectively on $c$ and $k$), so the total size of the alphabet is $O(ck)=O(1)$.

\section{Algorithm for Caterpillars on Small Alphabet}
\label{sec:algo}

In this section, we show that the hard instances of TED from Section~\ref{sec:apsp} can be solved
in time $O(n^{2}|\Sigma|^{2}\log n)$, where $n$ is the size of the trees and $\Sigma$ is the alphabet.
Recall that in such an instance we are given two trees $F$ and $G$ both consisting of a single path (called spine) of length $O(n)$ with a single leaf pending from every node, and all these leafs are to the right of the path in $F$ and to the left of the path in $G$ (see Figure~\ref{fig:macrogeneral}).
In the following we use the same notation as in Lemma~\ref{lem:redstructure}.
At a high level, we want to guess the rootmost non-spine node in the left tree $f'_{i_{p+1}}$ and 
the rootmost non-spine node in the right tree $g'_{j_{p+1}}$. The optimal matching of spine
nodes above these non-spine nodes can be precomputed in $O(n^{2})$ total time with a simple
dynamic programming algorithm. It might be tempting to say the same about the situation below,
but this is much more complicated due to the fact that leaf nodes in this part are matched in reversed order.
To overcome this difficulty, we need the following tool.

\begin{lemma}
\label{lem:reverse}
For strings $s[1..n]$ and $t[1..m]$ over alphabet $\Sigma$ and matching cost $\cmatch(c,d)$
for any two letters $c,d\in\Sigma$, we define the \emph{optimal matching of $s$ and the reverse of $t$} as 
\[ \min\Big\{ \sum_{\ell=1}^k \cmatch(s[i_\ell],t[j_\ell]) \Bigm\vert k \ge 0, \; 1 \le i_1 < \ldots < i_k \le n, \; 1 \le j_k < \ldots < j_1 \le m\Big\}. \]
Given two strings $s[1..n], t[1..n]$, in $O(n^{2}\log n)$ total time we can calculate, for every $i$ and $j$, the optimal matching of $s[1..i]$ and the reverse of $t[1..j]$.
\end{lemma}

\begin{proof}
We construct an $(n+1)\times(n+1)$ grid graph on nodes $v_{i,j}$, where $i,j=0,1,\ldots,n$ as follows.
For every $i,j=0,1,\ldots,n$, we create a directed edge from $v_{i,j}$ to $v_{i+1,j}$ and $v_{i,j+1}$ with length
zero. Also, we create a directed edge from $v_{i,j}$ to $v_{i+1,j+1}$ with length $\cmatch(s[i],t[n+1-j])$.
Then, paths from $v_{1,n+1-j}$ to $v_{i,n+1}$ are in one-to-one correspondence with matchings
of $s[1..i]$ to the reverse of $t[1..j]$. Therefore, the cheapest such path corresponds to the optimal
matching. 

The grid is a planar directed graph, and all starting nodes $v_{1,n+1-j}$ lie on the external
face, so we can use the multiple-source shortest paths algorithm of Klein~\cite{Klein05} to compute,
in $O(n^{2}\log n)$ time, a representation of shortest paths from all starting nodes $v_{1,n+1-j}$
to all nodes of the grid.\footnote{In the presence of negative-length edges, Klein's algorithm requires an initial shortest paths tree from some node on the external face to all other nodes. In our case, computing this initial shortest path tree can easily be done in $O(n^2)$ time as our graph is a directed acyclic graph.} This representation can be then queried in $O(\log n)$ time to extract
the length of any path from $v_{1,n+1-j}$ to $v_{i,n+1}$. Thus, the total time is $O(n^{2}\log n)$.
\end{proof}

To see how Lemma~\ref{lem:reverse} can be helpful, consider the (simpler) case when there are no
additional spine nodes $f_{i_{p+q+2}}$ and $g_{j_{p+q+2}}$. We construct two strings $s$ and $t$
by writing down the labels of leaf nodes $f'_{n},f'_{n-1},\ldots,f'_{1}$ and $g'_{n},g'_{n-1},\ldots,g'_{1}$,
respectively, and preprocess them using Lemma~\ref{lem:reverse}. Then, to find the optimal matching
we guess $i_{p+2}$ and $j_{p+2}$. As explained above, optimal matching of spine nodes above
$f'_{i_{p+2}}$ and $g'_{j_{p+2}}$ can be precomputed in $O(n^{2})$ time in advance. Then, we need
to match some of the leaf nodes $f'_{i_{p+2}},f'_{i_{p+2}+1},f'_{i_{p+2}+2},\ldots$ to some of the
leaf nodes $g'_{j_{p+2}},g'_{j_{p+2}+1},g'_{j_{p+2}+2},\ldots$ in the reversed order. This exactly
corresponds to matching $s[1,n+1-i_{p+2}]$ to the reverse of $t[1,n+1-j_{p+2}]$ and thus is also precomputed.
Iterating over all possible $i_{p+2}$ and $j_{p+2}$ gives us the optimal matching in $O(n^{2}\log n)$
total time.

Now consider the general case. We assume that both optional spine nodes $f_{i_{p+q+2}}$ and
$g_{j_{p+q+2}}$ exist; if only one of them is present the algorithm is very similar. As in the simpler case,
we iterate over all possible $i_{p+1}$ and $j_{p+1}$. The natural next step would be to iterate
over all possible $i_{p+q+2}$ and $j_{p+q+2}$, but this is too expensive. However, because
no spine nor leaf nodes below $f_{i_{p+q+2}}$ (or $g_{j_{p+q+2}}$) are matched, we can as well
replace $f_{i_{p+q+2}}$ with the lowest spine node with the same label (and similarly for $g_{j_{p+q+2}}$).
Thus, instead of guessing $i_{p+q+2}$ we can guess the label of $f_{i_{p+q+2}}$ and choose the
lowest spine node with such label (and similarly for $j_{p+q+2}$). 
Now we retrieve the precomputed optimal matching of spine
nodes above $f'_{i_{p+1}}$ and $g'_{j_{p+1}}$. Then we need to find the optimal matching of
leaf nodes $f'_{i_{p+1}+1},f'_{i_{p+1}+2},\ldots,f'_{i_{p+q+2}-1}$ and
$g'_{i_{p+1}+1},g'_{i_{p+1}+2},\ldots,g'_{j_{p+q+2}-1}$. This can be precomputed in $O(n^{2}|\Sigma|^{2} \log n)$
time with Lemma~\ref{lem:reverse}. Indeed, there are only $|\Sigma|$ possibilities for
$i_{p+q+2}-1$ and also $|\Sigma|$ possibilities for $j_{p+q+2}-1$, as both of them are defined by
the lowest occurrence of a label among the spine nodes of the left and the right tree, respectively.
For each such combination, we construct two strings $s$ and $t$ by writing down the labels of
leaf nodes above $f'_{p+q+2}$ and $g'_{p+q+2}$ in the bottom-up order and preprocess them
in $O(n^{2}\log n)$ time. This allows us to retrieve the optimal matching of leaf nodes and
then we only have to add $\cmatch(f'_{i_{p+1}},g_{j_{p+q+2}})$ and $\cmatch(f_{i_{p+q+2}},g'_{j_{p+1}})$
to obtain the total cost. Thus, after $O(n^{2}|\Sigma|^{2} \log n)$ preprocessing, we can find
the optimal matching by iterating over $n^{2}|\Sigma|^{2}$ possibilities and checking each of
them in constant time.

\bibliographystyle{abbrv}

\end{document}